\documentclass{article}
\usepackage{tikz, amsfonts, braket, amsmath, amssymb, amsthm, enumerate, subcaption,graphicx,cite,xcolor, setspace, lineno}
\usepackage[margin = 1in]{geometry}
\usetikzlibrary{decorations.markings,arrows.meta}
\usepackage[colorlinks = true]{hyperref}
\usepackage{soul}

\definecolor{lightblue}{RGB}{173,216,230}
\sethlcolor{lightblue}
\allowdisplaybreaks

\newtheorem{definition}{Definition}

\newtheorem{theorem}{Theorem}

\newtheorem{lemma}{Lemma}

\newtheorem{example}{Example}

\newtheorem{claim}{Claim}

\newtheorem{observation}{Observation}

\title{Quantum State Routing and Perfect State Transfer on Signed Graphs under Environmental Noise}
\author{Nur Mohammad Sanfui\thanks{Email: \texttt{nurmohammadsanfui@gmail.com}}, Supriyo Dutta\thanks{Email: \texttt{dosupriyo@gmail.com}}\\
	\small{Department of Mathematics}\\
	\small{National Institute of Technology Agartala}\\
	\small{Jirania, West Tripura, Tripura, India - 799046}
	}
\date{}
\begin{document}
	\maketitle

	\begin{abstract}
		Routing unknown quantum information across distributed communication networks requires autonomous, measurement-free mechanisms to prevent wave-function collapse. Szegedy quantum walks provide a mechanism for spatial state transport. The conventional walks on unweighted graphs suffer from severe back-reflection, spatial dispersion, and channel crosstalk. In this paper, we introduce a deterministic topological quantum routing architecture based on coined Szegedy quantum walks on edge-duplicated signed graphs. By treating edge signs as localized phase shifts within a balanced coin reflection, we enforce an exact zero back-scattering condition across routing nodes. We demonstrate deterministic Perfect State Transfer (PST) with unit fidelity at exact arrival times across fundamental archetypes, including the signed dumbbell switch $D_{2m,0,2n}$ and scalable glued binary trees. Furthermore, we analyze routing performance under realistic open-system amplitude and phase damping noise channels.
	\end{abstract}
	
	
	\tableofcontents
	
	\section{Introduction}
	
		The realization of scalable, distributed quantum networks relies fundamentally on the capability to route unknown photonic quantum states between arbitrary terminal nodes without measurement-induced decoherence \cite{bose2003quantum}. In classical communications, network routers actively inspect packet headers to direct traffic across physical crossbars. In the quantum regime, however, measuring a flying qubit causes wave-function collapse and violates the no-cloning theorem, precluding any header-decoding mechanism. Consequently, quantum routers must operate as autonomous, interferometric scattering fabrics that steer quantum payloads deterministically through coherent interference alone \cite{zueco2009quantum, paganelli2013routing, dutta2023quantum}.
		
		The concept of quantum walks \cite{portugal2013quantum} has been widely explored for quantum state transfer and routing. Spin-chain state transfer was established in \cite{bose2003quantum} as a baseline for unmodulated, coherent transport. Subsequent works formulated multi-terminal routers using minimal external control \cite{zueco2009quantum}, state-transfer scheduling on network subgraphs \cite{dutta2023quantum, dutta2026perfect}, and transitions to multi-user addressable routing \cite{paganelli2013routing}. In discrete-time quantum walks (DTQWs), local coin operations have been used to direct unknown states across lattices \cite{zhan2014perfect}, while spatial recurrence and arc-reversal periodicity have provided formal frameworks for localized network dynamics \cite{segawa2011localization, higuchi2017periodicity}. Concurrently, topological waveguide circuits have demonstrated that $\pi$-phase delays and lattice edge states can suppress unwanted scattering \cite{ozawa2019topological, li2022multiport}.
		
		Despite these advances, conventional quantum walks on standard unweighted networks suffer from severe transmission bottlenecks. At multi-branching junctions with vertex degrees $d \ge 3$, the unbiased coin reflection inevitably induces significant back-scattering and dispersive spreading across intermediate nodes. This diffusion destroys directional state refocusing and leads to unacceptable channel crosstalk, preventing deterministic Perfect State Transfer (PST) across non-trivial multi-port topologies.
		
		To overcome these structural limitations, we introduce an all-optical topological routing architecture based on coined Szegedy quantum walks on edge-duplicated signed graphs $\vec{\Sigma} = (V, \vec{E}, \sigma)$ \cite{brown2012perfect, harary1953notion}. The discrete edge signatures $\sigma(e) \in \{+1, -1\}$ physically indicate an integrated optical waveguides combined with localized $\pi$-phase shifters \cite{ozawa2019topological, li2022multiport}. Under a balanced coin parameter $\alpha = 0.5$, we establish an exact zero back-scattering condition $p = 1/2$ across routing junctions. This gauge-phase mechanism enforces destructive interference against unwanted back-reflections into input channels while constructively focusing $100\%$ of the transmission amplitude along designated egress paths.
		
		Using this framework, we construct and demonstrate two fundamental routing archetypes that achieve unit-fidelity state transfer ($F = 1.0$) at discrete, topologically deterministic arrival times:
		\begin{enumerate}
			\item \textbf{The Signed Dumbbell Switch ($D_{2m,0,2n}$):} A reconfigurable crossbar switch where toggling the sign of a single central bridge link alternates the network between a coherent storage loop (memory mode, $\sigma = +1$) and non-blocking forwarding to an antipodal terminal node (transmission mode, $\sigma = -1$) with arrival latency $\tau = m + 1 + n$.
			\item \textbf{The Scalable Glued-Tree Router:} A multi-path routing backbone formed by identifying the leaves of two depth-$d$ binary trees. By engineering alternating sign layers across generational levels, the network suppresses transverse diffusion and guides distributed wavepackets across the interface, achieving ballistic refocusing at the egress root at time $\tau = 2d$. Furthermore, we demonstrate that this architecture is structurally fault-tolerant against link failures.
		\end{enumerate}
		
		Finally, we analyze routing robustness under realistic open-system decoherence models. We derive exact analytical horizons exceeding the classical simulation threshold ($F_{\text{classical}} = 2/3$): while photon loss (amplitude damping) imposes an operational horizon of $\tau_{\max} \approx 0.4055/\lambda$ hops, pure dephasing (phase damping) exhibits an intrinsic multi-path mixture floor that preserves quantum advantage across $\tau_{\max} \approx 1.0986/p$ hops—more than double the distance of dissipative loss.
	
		A few concepts from graph theory \cite{west2001introduction} are essential to this article. A graph $G=(V(G), E(G))$ is a collection of two sets, which are the set of vertices $V(G)$ and the set of edges $E(G) \subseteq V(G) \times V(G)$. Throughout the article, the words vertex and node will be used synonymously. To draw a graph, we mark the vertices by dots and edges by lines joining two dots. An edge is called an undirected edge if it is not associated with a direction. An undirected edge between the vertices $u$ and $v$ is denoted by $e=(u,v)$. An edge with a direction is called a directed edge, or arc. In this article, we denote a directed edge from vertex $u$ to $v$ by $\overrightarrow{e}=\overrightarrow{(u,v)}$. Any undirected edge $(u, v)$ can be considered as a combination of two oppositely oriented directed edges $\overrightarrow{(u, v)}$ and $\overrightarrow{(v, u)}$. A signed undirected graph is a triple $\Sigma= (V(G), E(G), \sigma_{G})$, where $G=(V(G), E(G))$ is an undirected graph and $\sigma_{G}:E(G)\to\{+,-\}$ is a sign function assigning a sign to each edge. An edge with a positive sign is called a positive edge. Similarly, the sign of a negative edge is negative. A signed graph is depicted in \autoref{signedgraph}. To plot a signed graph, we use solid lines for positive edges and dashed lines for negative edges. A path between the vertices $v_0$ and $v_n$ is a sequence of vertices and edges $v_0,e_1,v_1,e_2,v_2,\dots,e_n,v_n$ where $e_i=(v_{i-1},v_i)$ for $i=1,2,\ldots,n$, where no vertex is repeated. The number of edges in a path is called the length of the path. A graph is said to be a connected graph if there exists a path between every pair of distinct vertices. The distance between two vertices is determined by the length of the shortest path between them.
	
		The remainder of this article is organized as follows: In Section~2, we formalize the abstract quantum router and its translation to signed graphs. In Section~3, we define the coined Szegedy quantum walk on signed graphs and prove the fundamental zero back-scattering lemmas. In Section~4, we present the fault-tolerant glued-tree router. In Section~5, we construct and analyze the signed dumbbell switch. Section~6 evaluates network performance under amplitude and phase damping noise channels. Finally, Section~7 concludes the paper.

	\section{Quantum routers and signed graphs}
	
		We initiate our discussion with the concept of a quantum router \cite{wei2022towards, azuma2023quantum}. The formal definition is as follows:
		\begin{definition}
			\textbf{(Quantum Router)}: A quantum router is a quantum dynamical system consisting of an input node (or port) $S$, a set of candidate destination nodes $\{D_1, D_2, \dots, D_k\}$, and a set of classical control parameters $\Theta_j$ such that an unknown quantum state $\vert{}\psi\rangle$ initialized at $S$ is transferred deterministically to a chosen target node $D_j$ at a specific arrival time $\tau_j$ with unit fidelity:
			$$U(\Theta_j)^{\tau_j} \vert{}S\rangle = e^{i\phi_j} \vert{}D_j\rangle, \quad \text{while } \langle D_i \vert{} U(\Theta_j)^{\tau_j} \vert{}S\rangle = 0,$$
			for all $i \neq j$ and $i, j = 1, 2, \dots ,k$ without measuring the state $\vert{}\psi\rangle$ or destroying its quantum coherence during transit.
		\end{definition}
		
		Therefore, to construct a quantum router we need a Network State Space ($\mathcal{H}_{\text{net}}$) which is a finite-dimensional Hilbert space representing the spatial transmission modes of the routing fabrics. We define an input port $\ket{S} \in \mathcal{H}_{\text{net}}$, which is the designated sender mode where an arbitrary quantum payload state $\vert{}\psi_{\text{in}}\rangle = c_0\vert{}0\rangle + c_1\vert{}1\rangle$ enters the system. The destination ports ($\{\ket{D_1}, \ket{D_2}, \dots, \ket{D_K}\} \subset \mathcal{H}_{\text{net}}$) form an orthonormal set of receiver modes, that is $\braket{D_i | D_j} = \delta_{ij}$. They represent the output terminals. The control configuration space $\mathcal{C}$ is a set of classical, reconfigurable control vectors $c = (c_1, c_2, \dots, c_M)$. Moreover, a discrete-time unitary evolution operator $U(\mathbf{c})$ acting on $\mathcal{H}_{\text{net}}$, parameterized by the control setting $\mathbf{c}$, governs state transmission.
		
		The system $\mathcal{R}$ is a deterministic quantum router if, for every target destination port $D_k$, there exists a specific control setting $\mathbf{c}^{(k)} \in \mathcal{C}$ and an integer transmission or arrival time $\tau_k \in \mathbb{N}$, such that the following three axioms hold:
		\begin{enumerate}
			\item 
				\textbf{Deterministic forwarding:} Deterministic forwarding or unit transfer of the quantum state is governed by $U(c^{(k)})^{\tau_k} \vert{}S\rangle = e^{i \theta_k} \vert{}D_k\rangle.$ It yields the  state-transfer fidelity $F_{S \to D_k}(\tau_k; c^{(k)}) = \vert{}\langle D_k \vert{} U(c^{(k)})^{\tau_k} \vert{}S\rangle\vert{}^2 = 1.0$. 
			\item 
				\textbf{Zero leakage:} The quantum state never leaks into unintended output ports. The condition for crosstalk or leakage suppression is determined by $\vert{}\langle D_j \vert{} U(c^{(k)})^{\tau_k} \vert{}S\rangle\vert{}^2 = 0$ for all $j \ne k$. 
			\item 
				\textbf{Coherence Preservation:} As the map $U(c^{(k)})^{\tau_k}$ is unitary, any arbitrary input superposition is delivered intact without measurement collapse, that is 
				$U(c^{(k)})^{\tau_k} \big(a \vert{}S\rangle \otimes \vert{}0\rangle + b \vert{}S\rangle \otimes \vert{}1\rangle \big) = e^{i\theta_k} \vert{}D_k\rangle \otimes \big(a \vert{}0\rangle + b \vert{}1\rangle \big).$
		\end{enumerate}
		
		Mathematically, we can translate the idea of quantum router to a signed graph. The input ports, destination ports, and intermediate junction or routing nodes of the network are represented by the vertices. The edges represent directed communication channels or physical waveguides connecting the vertices, duplicated as pairs of opposing directed arcs $\overrightarrow{(u, v)}$ and $\overrightarrow{(v, u)}$ for each link. In an integrated photonic waveguide circuit, a positive edge $\sigma = +1$ is a standard optical waveguide channel where light traverses without an induced relative phase change. A negative edge $\sigma = -1$ represents an optical channel with a localized $\pi$-phase shift $e^{i\pi} = -1$. In hardware, this is implemented using an electro-optic or thermo-optic phase modulator. Toggling the sign from $+1$ to $-1$ simply requires applying a voltage to introduce a $\pi$ phase delay. This phase shift directs the quantum interference pattern to forward the state along the target port rather than trapping it in a circulating loop. Therefore, a positive edge represents a normal, unmodified transmission link and a negative edge represents a localized $\pi$-phase shift that controls destructive interference and suppresses crosstalk.
		
		A topological router is a quantum routing architecture. Here, the direction, destination, and non-blocking delivery of a quantum state are determined by the global structural layout or topology and localized phase signatures of the network. Intermediate projective measurements and classical packet inspections are completely avoided. In conventional communication, a router reads an incoming packet's header and switches physical channels accordingly. In quantum networks, measuring a quantum packet collapses the wave-function and destroys quantum coherence. A topological router solves this by treating state routing as an interference-driven scattering process on an underlying graph. The information packet, for example a single photon, propagates as a coherent wavepacket via a quantum walk across the network's spatial modes.  By engineering the structural properties of graphs and the discrete gauge phases or signed edges, the router creates deterministic constructive interference toward the desired output port and complete destructive cancellation toward all unintended paths.  
		
		For simplicity, in all the quantum routers under our consideration, we assume that there is one input node $S$ and only one destination node or receiver node $R$. We consider a single control parameter $\alpha$ parameterizing the unitary operators that govern state transitions. The Szegedy quantum walks play a key role in information transition. We discuss it in the next section.

	\section{Szegedy quantum walk on signed graphs and state transfer}
	
		We begin with a few graph theoretic ideas which are essential in defining the Szegedy quantum walks \cite{szegedy2004quantum, szegedy2004spectra} on signed graphs. Let $e = (u, v)$ be an undirected edge. The vertices $u$ and $v$ are the endpoints of $e$. We say $v$ is a neighbor of $u$ or $v$ is adjacent to $u$. We denote the set of all neighbors of $u$ by $N(u)$. The degree of the vertex $u$ is $d(u) = |N(u)|$.
		
		Let $\overrightarrow{e} = \overrightarrow{(u, v)}$ be a directed edge. The vertex $u$ is called the origin of edge $\overrightarrow{e}$, and we denote $o(\overrightarrow{e}) = u$. Also, the vertex $v$ is called the terminal of $\overrightarrow{e}$ which is denoted by $t(\overrightarrow{e}) = v$. The inverse of a directed edge $\overrightarrow{e} = \overrightarrow{(u,v)}$ is denoted by $\overrightarrow{e}^{-1}$ and is defined by $\overrightarrow{e}^{-1}=\overrightarrow{(v,u)}$. Hence, $o(\overrightarrow{e}^{-1})= t(\overrightarrow{e})$ and $t(\overrightarrow{e}^{-1})= o(\overrightarrow{e})$.
		
		A signed directed graph, or signed digraph is an ordered triple $\overrightarrow{\Sigma} = (V(\overrightarrow{G}), E(\overrightarrow{G}), \sigma_{\overrightarrow{G}})$, where $\overrightarrow{G} = (V(\overrightarrow{G}), E(\overrightarrow{G}))$ is a directed graph and $\sigma_{\overrightarrow{G}}:E(\overrightarrow{G})\to\{+,-\}$ is a sign function assigning a sign to each edge. Given a signed directed graph $\overrightarrow{\Sigma}=(V(\overrightarrow{G}),E(\overrightarrow{G}),\sigma_{\overrightarrow{G}})$, the set of all outgoing edges from $u$ is denoted by $O_u = \{\overrightarrow{e} : o(\overrightarrow{e})=u\}$. The positive and negative out-degree of a vertex $u$ is the number of positive and negative outgoing edges from $u$ which are denoted by $d_{\mathrm{o}}^{+}(u)$ and $d_{\mathrm{o}}^{-}(u)$, respectively. Mathematically,
		\begin{equation}\label{outpositivedegree}
			d_{\mathrm{o}}^{+}(u)
			=
			\left|\left\{
			\overrightarrow{(u,v)}\in 	E(\overrightarrow{G}) :
			\sigma\overrightarrow{(u,v)}=+
			\right\}\right|,
		\end{equation}   
		and
		\begin{equation}\label{outnegativedegree}
			d_{\mathrm{o}}^{-}(u)
			=
			\left|\left\{
			\overrightarrow{(u,v)}\in 	E(\overrightarrow{G}) :
			\sigma\overrightarrow{(u,v)}=-
			\right\}\right|.
		\end{equation}
	
		To define the Szegedy quantum walks, we convert a signed undirected graph $\Sigma = (V(G), E(G), \sigma_{G})$ to a signed directed graph $\overrightarrow{\Sigma} = (V(\overrightarrow{G}), E(\overrightarrow{G}), \sigma_{\overrightarrow{G}})$, where $V(\overrightarrow{G}) = V(G)$ and for every undirected edge $(u,v)$ in $\Sigma$, we construct two oppositely oriented arcs $\overrightarrow{(u,v)}$ and $\overrightarrow{(v,u)}$ in $\overrightarrow{\Sigma}$ keeping the sign of edges unchanged. Clearly, $|E(\overrightarrow{G})| = 2 |E(G)|$. In \autoref{graph_and_its_duplication}, we depict an undirected signed graph and the directed signed graph after the duplication process.
	
		Suppose a wavepacket is located at vertex $v$. In the next step of quantum walks it will move to a neighboring vertex following the edges. Let the probability of selecting any edge of same sign is equal. Also we assume that the probabilities of choosing positive sign and negative sign are $\alpha$ and $(1 - \alpha)$, respectively.  Therefore, the probability of selecting an edge $\overrightarrow{e}$ with $o(\overrightarrow{e}) = v$ is 
		\begin{equation}\label{prob}
			p(\overrightarrow{e})=
			\begin{cases}
				\dfrac{1}{d_{\mathrm{o}}^{+}(v)}, & 	\text{if } d_{\mathrm{o}}^{-}(v)=0; \\
				\dfrac{1}{d_{\mathrm{o}}^{-}(v)}, & 	\text{if } d_{\mathrm{o}}^{+}(v)=0;\\
				\dfrac{\alpha}{d_{\mathrm{o}}^{+}(v)}, & 	\text{if } d_{\mathrm{o}}^{+}(v)\neq 0,\  d_{\mathrm{o}}^{-}(v)\neq 0 \  \text{and} \  \sigma(\overrightarrow{e}) = +; \\
				\dfrac{1-\alpha}{d_{\mathrm{o}}^{-}(v)}, 	& \text{if } d_{\mathrm{o}}^{+}(v)\neq 0,\  d_{\mathrm{o}}^{-}(v)\neq 0 \  \text{and} \  \sigma(\overrightarrow{e}) = -.
			\end{cases}
		\end{equation}
	
		We label the vertices in $\overrightarrow{G}$ with $0, 1, 2, \dots, (|V(G)|-1)$ and arrange the directed edges in dictionary order on their end vertices as  $\overrightarrow{e_0}, \overrightarrow{e_1}, \overrightarrow{e_2}, \dots, \overrightarrow{e_{(2|E(G)|-1)}}$. Corresponding to an edge $\overrightarrow{e_i}$, we assign a state vector 
	\begin{equation}
		\ket{\overrightarrow{e_i}} = (0,0, \dots,1(i\text{-th position}),0, \dots,0)^\dagger
	\end{equation} 
	for $i = 0, 1, 2, \dots, (2|E(G)|-1)$. The corresponding stochastic matrix of random walk is $M = (m_{ij})_{|V(G)|\times |V(G)|}$, where
	\begin{equation}
		m_{ij} = 
		\begin{cases}
			p(\overrightarrow{(j,i)}), & \text{if}\  \overrightarrow{(j, i)} \in E(\overrightarrow{G}); \\
			0, & \text{otherwise},
		\end{cases}
	\end{equation}
	where $p(\overrightarrow{(j,i)})$ is defined in equation \eqref{prob}.	
	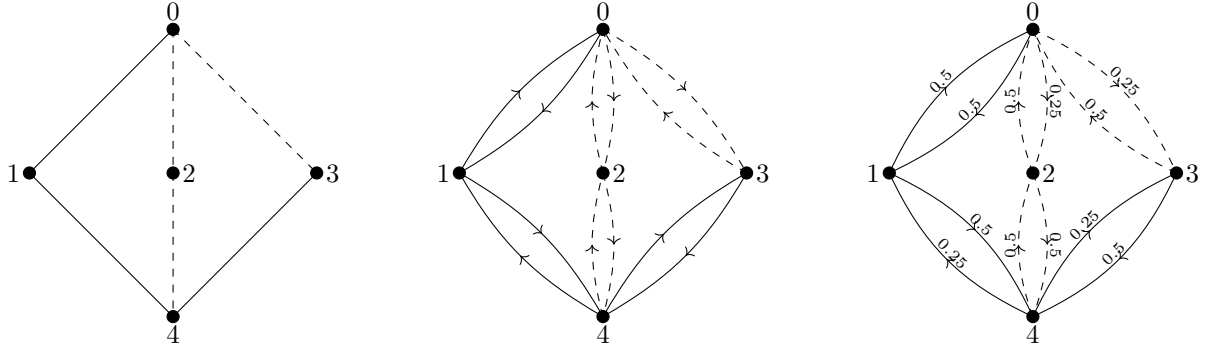
\begin{figure}
		\centering
		\begin{subfigure}[t]{0.3\textwidth}
			\centering
			\begin{tikzpicture}[x=1.9cm,y=1.9cm]
				
				\draw[fill] (0,1) circle[radius=0.08cm];
				\node[above] at (0,1) {$0$};
				
				\draw[fill] (-1,0) circle[radius=0.08cm];
				\node[left] at (-1,0) {$1$};
				
				\draw[fill] (0,0) circle[radius=0.08cm];
				\node[right] at (0,0) {$2$};
				
				\draw[fill] (1,0) circle[radius=0.08cm];
				\node[right] at (1,0) {$3$};
				
				\draw[fill] (0,-1) circle[radius=0.08cm];
				\node[below] at (0,-1) {$4$};
				
				\draw (0,1)--(-1,0);
				\draw (-1,0)--(0,-1);
				\draw (1,0)--(0,-1);
				
				\draw[dashed] (0,1)--(0,0);
				\draw[dashed] (0,1)--(1,0);
				\draw[dashed] (0,0)--(0,-1);
				
			\end{tikzpicture}
			\caption{Here is a signed undirected graph with five vertices with three positive and three negative edges. The positive edges are drawn by solid lines and the negative edges are drawn with dashed line.}
			\label{signedgraph}
		\end{subfigure}
		\hspace{.5cm}
		\begin{subfigure}[t]{0.3\textwidth}
			\centering
			\begin{tikzpicture}[
				x=1.9cm,y=1.9cm,
				midpositive/.style={
					postaction={decorate},
					decoration={markings, mark=at position 0.5 with {\arrow{>}}}
				},
				midnegative/.style={
					dashed,
					postaction={decorate},
					decoration={markings, mark=at position 0.5 with {\arrow{>}}}
				}
				]
				
				\draw[fill] (0,1) circle[radius=0.08cm];
				\node[above] at (0,1) {$0$};
				
				\draw[fill] (-1,0) circle[radius=0.08cm];
				\node[left] at (-1,0) {$1$};
				
				\draw[fill] (0,0) circle[radius=0.08cm];
				\node[right] at (0,0) {$2$};
				
				\draw[fill] (1,0) circle[radius=0.08cm];
				\node[right] at (1,0) {$3$};
				
				\draw[fill] (0,-1) circle[radius=0.08cm];
				\node[below] at (0,-1) {$4$};
				
				\draw[midpositive,bend left=15] (0,1) to (-1,0);
				\draw[midpositive,bend left=15] (-1,0) to (0,1);
				
				\draw[midpositive,bend left=15] (-1,0) to (0,-1);
				\draw[midpositive,bend left=15] (0,-1) to (-1,0);
				
				\draw[midpositive,bend left=15] (1,0) to (0,-1);
				\draw[midpositive,bend left=15] (0,-1) to (1,0);
				
				\draw[midnegative,bend left=15] (0,1) to (0,0);
				\draw[midnegative,bend left=15] (0,0) to (0,1);
				
				\draw[midnegative,bend left=15] (0,1) to (1,0);
				\draw[midnegative,bend left=15] (1,0) to (0,1);
				
				\draw[midnegative,bend left=15] (0,0) to (0,-1);
				\draw[midnegative,bend left=15] (0,-1) to (0,0);
				
			\end{tikzpicture}
			\caption{We apply two opposite orientations on every edge of the graph in \autoref{signedgraph}. Corresponding to a positive edge we have two positive directed edges. Also, for any negative edge we get two negative directed edges.}
		\end{subfigure}
		\hspace{.5cm}
		\begin{subfigure}[t]{0.3\textwidth}
			\centering
			\begin{tikzpicture}[
				x=1.9cm,y=1.9cm,
				midpositive/.style={
					postaction={decorate},
					decoration={markings, mark=at position 0.5 with {\arrow{>}}}
				},
				midnegative/.style={
					dashed,
					postaction={decorate},
					decoration={markings, mark=at position 0.5 with {\arrow{>}}}
				},
				plabel/.style={
					font=\fontsize{6}{7}\selectfont,
					sloped, above, pos=0.5, inner sep=1pt
				}
				]
				
				\draw[fill] (0,1) circle[radius=0.08cm];
				\node[above] at (0,1) {$0$};
				
				\draw[fill] (-1,0) circle[radius=0.08cm];
				\node[left] at (-1,0) {$1$};
				
				\draw[fill] (0,0) circle[radius=0.08cm];
				\node[right] at (0,0) {$2$};
				
				\draw[fill] (1,0) circle[radius=0.08cm];
				\node[right] at (1,0) {$3$};
				
				\draw[fill] (0,-1) circle[radius=0.08cm];
				\node[below] at (0,-1) {$4$};
				
				\draw[midpositive,bend left=20] (0,1)  to node[plabel]{$0.5$}  (-1,0);
				\draw[midpositive,bend left=20] (-1,0) to node[plabel]{$0.5$}  (0,1);
				\draw[midpositive,bend left=20] (-1,0) to node[plabel]{$0.5$}  (0,-1);
				\draw[midpositive,bend left=20] (0,-1) to node[plabel]{$0.25$} (-1,0);
				\draw[midpositive,bend left=20] (1,0)  to node[plabel]{$0.5$}  (0,-1);
				\draw[midpositive,bend left=20] (0,-1) to node[plabel]{$0.25$} (1,0);
				
				\draw[midnegative,bend left=20] (0,1)  to node[plabel]{$0.25$} (0,0);
				\draw[midnegative,bend left=20] (0,0)  to node[plabel]{$0.5$}  (0,1);
				\draw[midnegative,bend left=20] (0,1)  to node[plabel]{$0.25$} (1,0);
				\draw[midnegative,bend left=20] (1,0)  to node[plabel]{$0.5$}  (0,1);
				\draw[midnegative,bend left=20] (0,0)  to node[plabel]{$0.5$}  (0,-1);
				\draw[midnegative,bend left=20] (0,-1) to node[plabel]{$0.5$}  (0,0);
				
			\end{tikzpicture}
			\caption{Here we show the edge duplicated graph with probability on edges beside the arrow, when $\alpha = 0.5$.}
		\end{subfigure}
		\caption{A signed undirected graph and a signed directed graph generated by assigning two opposite orientations on every edge.}
		\label{graph_and_its_duplication}
	\end{figure}
	
	We define an incidence vector \cite{segawa2011localization} $a_u$ corresponding to vertex $u$ which is $a_u=\sum_{\overrightarrow{e} \in O_u}\sqrt{p(\overrightarrow{e})}\ket{\overrightarrow{e}}$.
	The Hilbert space $\mathcal{H}$ of the quantum walk is $\mathcal{H} = \operatorname{span}\{\ket{\overrightarrow{e}}:\overrightarrow{e} \in E(\overrightarrow{G})\}.$ The evolution matrix of the quantum walk, $U_\alpha=SC_\alpha$ where $S$ is the flip-flop shift operator $S\ket{\overrightarrow{(u,v)}}=\ket{\overrightarrow{(v,u)}}$ for all $\overrightarrow{(u,v)} \in E(\overrightarrow{G})$ and 
	\begin{equation}\label{coin}
		C_\alpha=2\left(\sum_{u}a_ua_u^\dagger\right)-I_{|E(\overrightarrow{G})|}.
	\end{equation}
	is the coin operator. The evolution matrix \cite{higuchi2017periodicity}, $U_\alpha$ is of order $2|E(G)| \times 2|E(G)|$ which is indexed by the arcs. 
	\begin{equation}\label{szegedy matrix}
		(U_\alpha)_{\overrightarrow{f},\overrightarrow{e}}=
		\begin{cases}
			2\sqrt{p(\overrightarrow{e})}\sqrt{p(\overrightarrow{f}^{-1})}-\delta_{\overrightarrow{e}^{-1},\overrightarrow{f}},
			& \text{if } t(\overrightarrow{f})=o(\overrightarrow{e}), \\
			0, & \text{otherwise}.
		\end{cases}
	\end{equation}
	Let $\ket{\overrightarrow{e}}$ be a vector associated with an edge $\overrightarrow{e}$, then
	\begin{equation}
		U_\alpha\ket{\overrightarrow{e}}=  \sum_{\substack{\overrightarrow{f}: t(\overrightarrow{f})=o(\overrightarrow{e})}} (U_\alpha)_{\overrightarrow{f},\overrightarrow{e}} \ket{\overrightarrow{f}}
		\label{evolution}
	\end{equation}
	
	Note that the quantum walk operator $(U_\alpha)$ depends on the parameter $\alpha$. It breaks down the quantum walk dynamics into two physical stages. The coin reflection $C$ occurs locally inside each router junction or vertex $u$ to redistribute the probability amplitudes among its outgoing physical channels. The shift operator $S$ transfers the packets across the physical waveguide link to the adjacent junctions.

		Now we are in a position to define the perfect state transfer between two vertices $u$ and $v$ in a graph \cite{bose2003quantum, christandl2004perfect, chapman2016experimental}.	
		\begin{definition}\label{pst}
			\textbf{(Perfect state transfer)} There is a perfect state transfer from vertex $u$ to vertex $v$ at time $t$ if there exists a complex scalar $\lambda$ with $|\lambda|=1$, such that $(U_\alpha)^t \ket{u} = \lambda \ket{v}$, where $\ket{u}=\frac{1}{\sqrt{|O_u|}}\sum_{\overrightarrow{e} \in O_u}\ket{\overrightarrow{e}}$ and $\ket{v}$ is defined similarly.
		\end{definition}
		Clearly, if there is a PST from $u$ to $v$, then the fidelity between $(U_\alpha)^t \ket{u}$ and $\ket{v}$ at time $t$ with respect to the parameter $\alpha$ becomes
		\begin{equation}
			F_{u\rightarrow v}(t;\alpha) = |\braket{v | (U_\alpha)^t | u}|^2 = |\lambda|^2 = 1,
		\end{equation}
		which we also consider an alternative definition of PST in our numerical experiments. 
		
		Perfect State Transfer is crucial in quantum routing. It guarantees the deterministic transport of an arbitrary quantum payload to a designated target port with unit fidelity $F = 1.0$ at a discrete arrival time determined by the graph topology. Also the transmission in PST is completely measurement-free. Therefore the routing of quantum packets via PST does not collapse the quantum superposition or destroy entanglement. 
		
		Now we have the following lemmas, which are useful in the next sections to prove PST between vertices in a signed graph.

		\begin{lemma}\label{half lemma}
			Let $\Sigma= (V(G), E(G), \sigma_{G})$ be a signed graph and $\overrightarrow{e} = \overrightarrow{(u, v)} \in E(\overrightarrow{G})$. Then $p(\overrightarrow{e})$ = $\frac{1}{2}$ if and only if $(U_\alpha)_{\overrightarrow{e}^{-1}, \overrightarrow{e}} = 0$.
		\end{lemma}
		
		\begin{proof}
			Note that $(U_\alpha)_{\overrightarrow{e}^{-1},\overrightarrow{e}} = (U_\alpha)_{\overrightarrow{(v,u)},\overrightarrow{(u,v)}} = 2 \sqrt{p(\overrightarrow{(u,v)})} \sqrt{p(\overrightarrow{(v,u)}^{-1})} - \delta_{\overrightarrow{(u,v)}^{-1}, \overrightarrow{(v,u)}}$ from equation\eqref{szegedy matrix}. Now, $(U_\alpha)_{\overrightarrow{e}} = 0$ if and only if $2 \sqrt{p(\overrightarrow{(u,v)})} \sqrt{p(\overrightarrow{(u,v)})} = \delta_{\overrightarrow{(v,u)}, \overrightarrow{(v,u)}}$. Here, $\delta_{\overrightarrow{(v,u)}, \overrightarrow{(v,u)}} = 1$ if and only if $\overrightarrow{(u,v)} = \overrightarrow{(v,u)}$. Applying it we get $p(\overrightarrow{e})=p(\overrightarrow{(u,v)}) = \frac{1}{2}$. Hence, the proof.
		\end{proof}
	
	\begin{lemma}\label{backward moving}
		Let $\Sigma= (V(G), E(G), \sigma_{G})$ be a signed graph and $v\in V(G)$ has only two neighbors $u$ and $w$, then $U_{0.5}\ket{\overrightarrow{(v,w)}}=\ket{\overrightarrow{(u,v)}}.$ 
	\end{lemma}
	
		\begin{proof}
			As two edges are incident to $v$, there may be four different combinations of $+$ and $-$ signs on them. 
			Also,  $O_v=\{\overrightarrow{(v,u)},\overrightarrow{(v,w)}\}$ and 
			\begin{equation}\label{two arc}
				\{\overrightarrow{f}\in E(\overrightarrow{G}): t(\overrightarrow{f})=v\}=\{\overrightarrow{(u,v)},\overrightarrow{(w,v)}\}.
			\end{equation}
			Applying equation (\ref{two arc}) and  $\alpha = 0.5$ in equation (\ref{prob}), we get $p(\overrightarrow{(v,u)})=p(\overrightarrow{(v,w)})=\frac{1}{2}$. Now, applying equation\eqref{evolution}, we get 
			$$U_{0.5}\ket{\overrightarrow{(v,w)}} = \sum_{\substack{\overrightarrow{f}: t(\overrightarrow{f})=o(\overrightarrow{(v,w)})}} (U_{0.5})_{\overrightarrow{f},\overrightarrow{(v,w)}} \ket{\overrightarrow{f}} = \sum_{\substack{\overrightarrow{f}: t(\overrightarrow{f})=v}} (U_{0.5})_{\overrightarrow{f},\overrightarrow{(v,w)}} \ket{\overrightarrow{f}}.$$
			From equation\eqref{two arc} we get 
			\[
			\begin{aligned}
				U_{0.5}\ket{\overrightarrow{(v,w)}}&=(U_{0.5})_{\overrightarrow{(u,v)},\overrightarrow{(v,w)}} \ket{\overrightarrow{(u,v)}} + (U_{0.5})_{\overrightarrow{(w,v)},\overrightarrow{(v,w)}} \ket{\overrightarrow{(w,v)}}  \\
				&=(2\sqrt{p(\overrightarrow{v,w})}\sqrt{p(\overrightarrow{(u,v)}^{-1}}-0)\ket{\overrightarrow{(u,v)}} + 0\times \ket{\overrightarrow{(w,v)}}	~\text{(from equation\eqref{szegedy matrix} and lemma~\eqref{half lemma})} \\	
				&=(2\sqrt{p(\overrightarrow{v,w})}\sqrt{p(\overrightarrow{(v,u)}})\ket{\overrightarrow{(u,v)}} =(2\sqrt{\frac{1}{2}}\sqrt{\frac{1}{2}})\ket{\overrightarrow{u,v}} =\ket{\overrightarrow{u,v}}.
			\end{aligned}
			\]  
			Hence, the proof.
		\end{proof}

	\section{Construction of a fault-tolerant quantum router}
		
		In this section, we discuss a quantum router whose topology mathematically depends on the glued trees. In a graph, a cycle is a closed path in which the first and last vertices coincide and no other vertex is repeated. A acyclic connected graph is called a tree. A rooted tree is a tree in which one vertex is designated as the root. Let $T = (V(T), E(T))$ be a rooted tree with root $r$. If $d(u)=1$ for some $u\in V(G)\setminus{r}$, then $u$ is called a leaf of $T$. The set of all leaves of $T$ is denoted by $L(T)$. In a rooted tree with root $r$, the depth of a vertex $v$ is denoted by $depth(v)$ which is the length of the path from $r$ to $v$. Therefore, $depth(r)=0$. For an edge $e=(u,v)$ where $u$ is the parent of $v$, the level of $e$ is $\ell(e)$ = $depth(v)$. Thus, the edges incident to $r$ have level $1$, and $\ell$ increases by $1$ at each step away from the root.
		
		Now we define two constructions of glued trees. 
		\begin{definition}\label{rooted tree definition}
			\textbf{(Rooted glued signed tree by leaf identification)} Let $\Sigma=(V(T),E(T),\sigma_T)$ be a rooted signed tree with root $r$ and leaf set $L(T)$. Consider two disjoint copies of $\Sigma$, denoted by
			
			$$
			\Sigma_1=(V(T_1),E(T_1),\sigma_{T_1})
			\quad\text{and}\quad
			\Sigma_2=(V(T_2),E(T_2),\sigma_{T_2}),
			$$
			with roots $r_1$ and $r_2$, respectively. A rooted glued signed tree by leaf identification is the signed graph obtained by identifying each leaf of $\Sigma_1$ with its corresponding leaf in $\Sigma_2$. The roots of the new graph are $r_1$ and $r_2$.
		\end{definition}
		
		\begin{claim}
			Let $T$ be a rooted tree such that 
			\begin{enumerate}
				\item Distance of every leaf from the root is equal.
				\item Edges at the same level have the same sign.
				\item Edges at every level alternates its sign from previous level.
			\end{enumerate}
			Then the glued tree generated by $T$ admits PST between its roots.
		\end{claim}
		
		\begin{example}
			Consider a tree depicted in \autoref{fig:original_tree}. It has the leaves $3, 4, 5$, and $6$. We take its duplicate in \autoref{fig:duplicate_tree}. The leaves in this graph are $10, 11, 12$, and $13$, which are identical to $3, 4, 5$, and $6$, respectively. We glue the identical leaf vertices and keep one leaf for each of the pairs of leaves $(3, 10)$, $(4, 11)$, $(5, 12)$, and $(6, 13)$. The new graph is depicted in  \autoref{fig:glued_tree_leaf_identification}, which we call a rooted glued tree. The root vertices are $0$ and $7$.
			
			\begin{figure}
				\centering
				\begin{subfigure}[b]{0.22\textwidth}
					\centering
					\begin{tikzpicture}[scale=0.7]
						
						\draw [fill] (0,0) circle [radius=0.08];
						\node [above] at (0,0) {$0$};
						
						\draw [fill] (-1.5,-1.5) circle [radius=0.08];
						\node [left] at (-1.5,-1.5) {$1$};
						\draw [fill] (1.5,-1.5) circle [radius=0.08];
						\node [right] at (1.5,-1.5) {$2$};
						
						\draw [fill] (-2.2,-3) circle [radius=0.08];
						\node [below] at (-2.2,-3) {$3$};
						\draw [fill] (-0.8,-3) circle [radius=0.08];
						\node [below] at (-0.8,-3) {$4$};
						\draw [fill] (0.8,-3) circle [radius=0.08];
						\node [below] at (0.8,-3) {$5$};
						\draw [fill] (2.2,-3) circle [radius=0.08];
						\node [below] at (2.2,-3) {$6$};
						
						\draw (0,0) -- (-1.5,-1.5);
						\draw (0,0) -- (1.5,-1.5);
						\draw [dashed] (-1.5,-1.5) -- (-2.2,-3);
						\draw [dashed] (-1.5,-1.5) -- (-0.8,-3);
						\draw [dashed] (1.5,-1.5) -- (0.8,-3);
						\draw [dashed] (1.5,-1.5) -- (2.2,-3);
						
					\end{tikzpicture}
					\caption{The rooted signed tree with root $0$.}
					\label{fig:original_tree}
				\end{subfigure}
				\hfill
				\begin{subfigure}[b]{0.22\textwidth}
					\centering
					\begin{tikzpicture}[scale=0.7]
						
						\draw [fill] (0,0) circle [radius=0.08];
						\node [above] at (0,0) {$7$};
						
						\draw [fill] (-1.5,-1.5) circle [radius=0.08];
						\node [left] at (-1.5,-1.5) {$8$};
						\draw [fill] (1.5,-1.5) circle [radius=0.08];
						\node [right] at (1.5,-1.5) {$9$};
						
						\draw [fill] (-2.2,-3) circle [radius=0.08];
						\node [below] at (-2.2,-3) {$10$};
						\draw [fill] (-0.8,-3) circle [radius=0.08];
						\node [below] at (-0.8,-3) {$11$};
						\draw [fill] (0.8,-3) circle [radius=0.08];
						\node [below] at (0.8,-3) {$12$};
						\draw [fill] (2.2,-3) circle [radius=0.08];
						\node [below] at (2.2,-3) {$13$};
						
						\draw (0,0) -- (-1.5,-1.5);
						\draw (0,0) -- (1.5,-1.5);
						\draw [dashed] (-1.5,-1.5) -- (-2.2,-3);
						\draw [dashed] (-1.5,-1.5) -- (-0.8,-3);
						\draw [dashed] (1.5,-1.5) -- (0.8,-3);
						\draw [dashed] (1.5,-1.5) -- (2.2,-3);
						
					\end{tikzpicture}
					\caption{A copy of the tree in \autoref{fig:original_tree} with root $7$.}
					\label{fig:duplicate_tree}
				\end{subfigure}
				\hfill
				\begin{subfigure}[b]{0.5\textwidth}
					\centering
					\begin{tikzpicture}[scale=0.5]
						
						\draw [fill] (0,0) circle [radius=0.08];
						\node [above] at (0,0) {$0$};
						
						\draw [fill] (-1.8,-1.5) circle [radius=0.08];
						\node [above] at (-1.8,-1.5) {$1$};
						\draw [fill] (1.8,-1.5) circle [radius=0.08];
						\node [above] at (1.8,-1.5) {$2$};
						
						\draw [fill] (-3.2,-3) circle [radius=0.08];
						\node [above] at (-3.2,-3) {$3$};
						\draw [fill] (-1.0,-3) circle [radius=0.08];
						\node [above] at (-1.0,-3) {$4$};
						\draw [fill] (1.0,-3) circle [radius=0.08];
						\node [above] at (1.0,-3) {$5$};
						\draw [fill] (3.2,-3) circle [radius=0.08];
						\node [above] at (3.2,-3) {$6$};
						
						\draw [fill] (-1.8,-4.5) circle [radius=0.08];
						\node [above] at (-1.8,-4.5) {$8$};
						\draw [fill] (1.8,-4.5) circle [radius=0.08];
						\node [above] at (1.8,-4.5) {$9$};
						
						\draw [fill] (0,-6) circle [radius=0.08];
						\node [above] at (0,-6) {$7$};
						
						\draw (0,0) -- (-1.8,-1.5);
						\draw (0,0) -- (1.8,-1.5);
						\draw [dashed] (-1.8,-1.5) -- (-3.2,-3);
						\draw [dashed] (-1.8,-1.5) -- (-1.0,-3);
						\draw [dashed] (1.8,-1.5) -- (1.0,-3);
						\draw [dashed] (1.8,-1.5) -- (3.2,-3);
						
						\draw (0,-6) -- (-1.8,-4.5);
						\draw (0,-6) -- (1.8,-4.5);
						\draw [dashed] (-1.8,-4.5) -- (-3.2,-3);
						\draw [dashed] (-1.8,-4.5) -- (-1.0,-3);
						\draw [dashed] (1.8,-4.5) -- (1.0,-3);
						\draw [dashed] (1.8,-4.5) -- (3.2,-3);
						
					\end{tikzpicture}
					\caption{Two disjoint copies the graph in \autoref{fig:original_tree} is taken in \autoref{fig:original_tree} and \autoref{fig:duplicate_tree}. There are alternating positive and negative edges at different levels. We glue the leaves to generate the rooted signed glued tree.}
					\label{fig:glued_tree_leaf_identification}
				\end{subfigure}
				\\
				\vspace{.5cm}

				\begin{subfigure}[t]{0.45\textwidth}
					\centering
					\includegraphics[width=\textwidth]{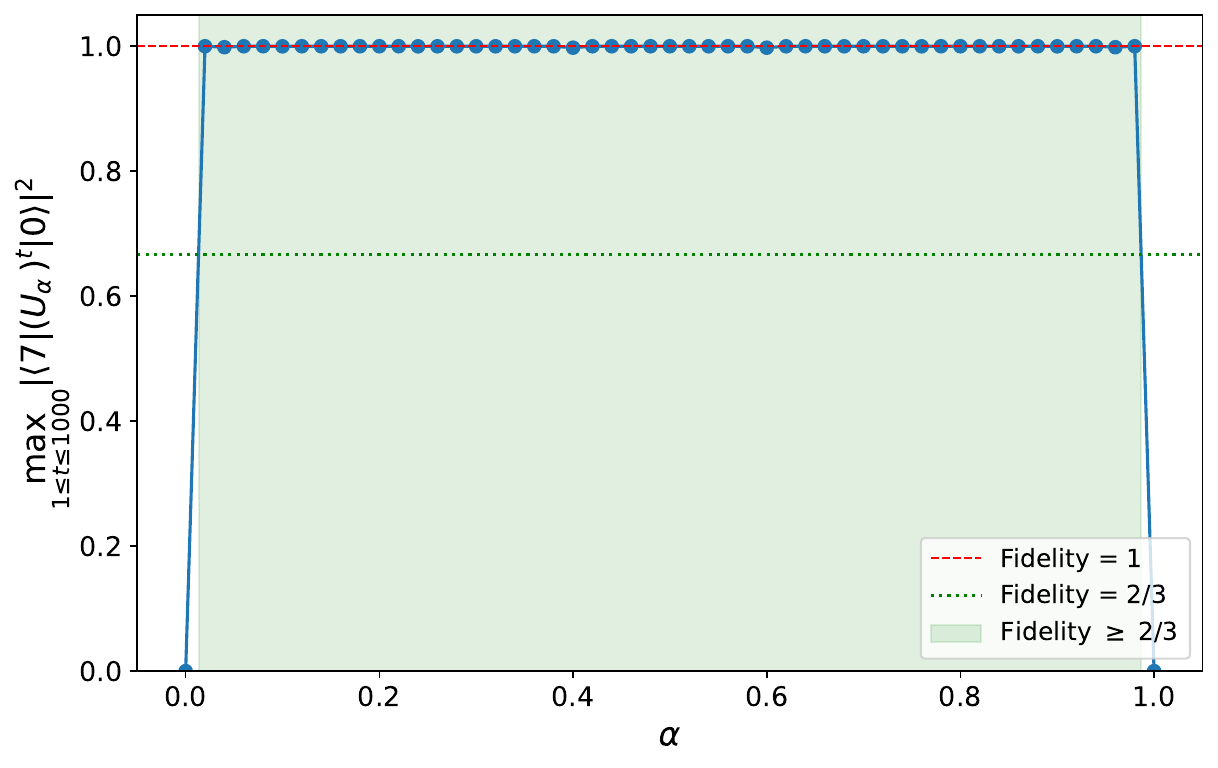}
					\caption{Here we plot the maximum fidelity from vertex $0$ to $7$ with respect to $\alpha$. PST occurs only for $\alpha=0.5.$ We observe that the fidelity is more than $2/3$ when $\alpha \in (0, 1)$.}
				\end{subfigure}
				\hspace{0.5cm}
				\begin{subfigure}[t]{0.45\textwidth}
					\centering
					\includegraphics[width=\textwidth]{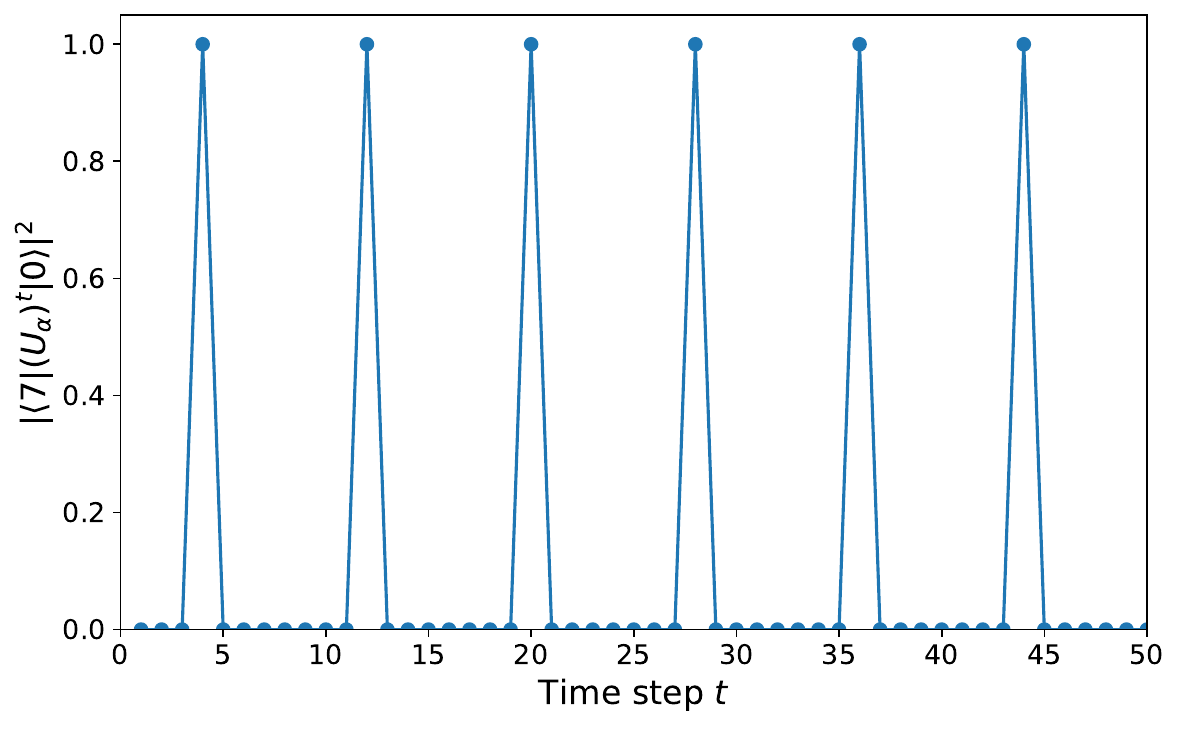}
					\caption{For $\alpha=0.5$, we calculate the fidelity to transfer a state from $0$ to $7$. The Fidelity is $1$ at $t = 4, 12, 20$, etc.}
				\end{subfigure}
				\caption{We study the dependency of fidelity on the parameter $\alpha$ for the graph depicted in \autoref{fig:glued_tree_leaf_identification}.}
			\end{figure}
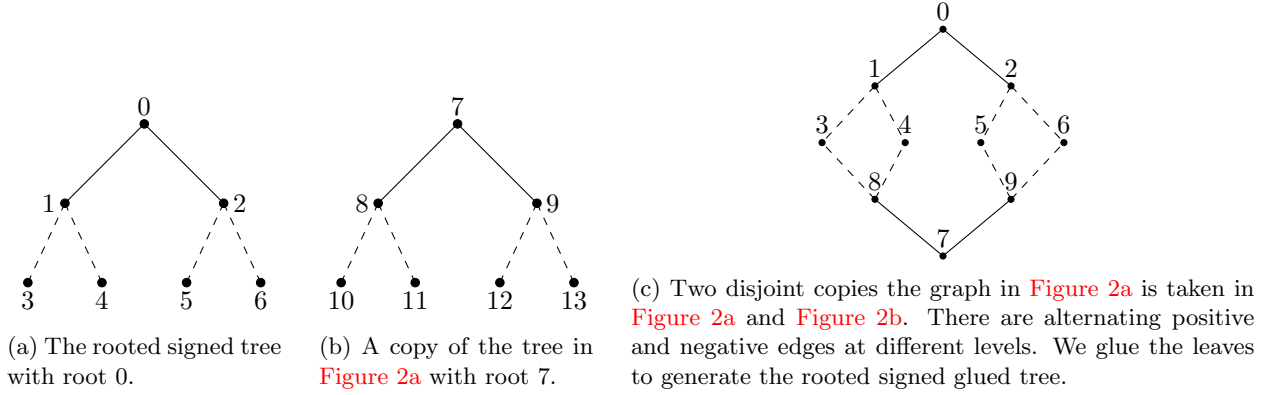
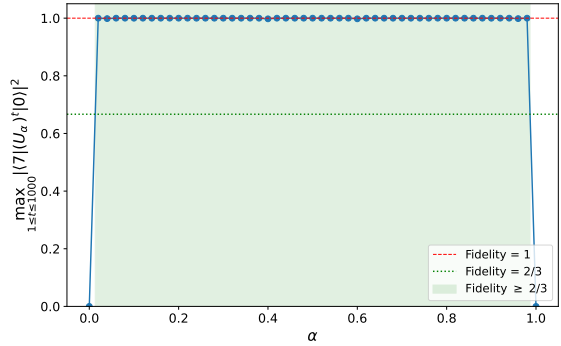
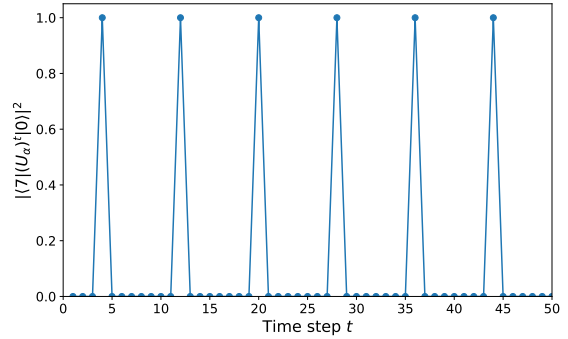
		\end{example}
	
		Now we have the following numerical observations:
		\begin{observation}
			For every value of $\alpha$ in (0,1) the fidelity of state transfer between $0$ and $7$ is more than $0.99$. The PST occurs for $\alpha=0.5.$  
		\end{observation}
		
		\begin{observation}
			We get PST for $\alpha=0.5$ at $t=4, 12,20,\ldots$. Note that the distance between $0$ to $7$ in the glued tree is 4.
		\end{observation}

		\begin{figure}[p]
			\centering
			\begin{subfigure}[t]{0.48\textwidth}
				\centering
				\begin{tikzpicture}[every node/.style={font=\footnotesize, inner sep=1.5pt}]
					
					\draw [fill] (0,0) circle [radius=0.05];
					\node [above] at (0,0) {$0$};
					
					\draw [fill] (-1.4,-0.65) circle [radius=0.05];
					\node [above] at (-1.4,-0.65) {$1$};
					\draw [fill] (1.4,-0.65) circle [radius=0.05];
					\node [above] at (1.4,-0.65) {$2$};
					
					\draw [fill] (-2.1,-1.3) circle [radius=0.05];
					\node [above] at (-2.1,-1.3) {$3$};
					\draw [fill] (-0.7,-1.3) circle [radius=0.05];
					\node [above] at (-0.7,-1.3) {$4$};
					\draw [fill] (0.7,-1.3) circle [radius=0.05];
					\node [above] at (0.7,-1.3) {$5$};
					\draw [fill] (2.1,-1.3) circle [radius=0.05];
					\node [above] at (2.1,-1.3) {$6$};
					
					\draw [fill] (-2.3,-1.95) circle [radius=0.05];
					\node [left] at (-2.3,-1.95) {$7$};
					\draw [fill] (-1.8,-1.95) circle [radius=0.05];
					\node [right] at (-1.8,-1.95) {$8$};
					\draw [fill] (-1.0,-1.95) circle [radius=0.05];
					\node [left] at (-1.0,-1.95) {$9$};
					\draw [fill] (-0.5,-1.95) circle [radius=0.05];
					\node [right] at (-0.5,-1.95) {$10$};
					\draw [fill] (0.5,-1.95) circle [radius=0.05];
					\node [left] at (0.5,-1.95) {$11$};
					\draw [fill] (1.0,-1.95) circle [radius=0.05];
					\node [right] at (1.0,-1.95) {$12$};
					\draw [fill] (1.8,-1.95) circle [radius=0.05];
					\node [left] at (1.8,-1.95) {$13$};
					\draw [fill] (2.3,-1.95) circle [radius=0.05];
					\node [right] at (2.3,-1.95) {$14$};
					
					\draw [fill] (-2.1,-2.6) circle [radius=0.05];
					\node [below] at (-2.1,-2.6) {$18$};
					\draw [fill] (-0.7,-2.6) circle [radius=0.05];
					\node [below] at (-0.7,-2.6) {$19$};
					\draw [fill] (0.7,-2.6) circle [radius=0.05];
					\node [below] at (0.7,-2.6) {$20$};
					\draw [fill] (2.1,-2.6) circle [radius=0.05];
					\node [below] at (2.1,-2.6) {$21$};
					
					\draw [fill] (-1.4,-3.25) circle [radius=0.05];
					\node [below] at (-1.4,-3.25) {$16$};
					\draw [fill] (1.4,-3.25) circle [radius=0.05];
					\node [below] at (1.4,-3.25) {$17$};
					
					\draw [fill] (0,-3.9) circle [radius=0.05];
					\node [below] at (0,-3.9) {$15$};

					\draw (0,0) -- (-1.4,-0.65);
					\draw (0,0) -- (1.4,-0.65);
					\draw [dashed] (-1.4,-0.65) -- (-2.1,-1.3);
					\draw [dashed] (-1.4,-0.65) -- (-0.7,-1.3);
					\draw [dashed] (1.4,-0.65) -- (0.7,-1.3);
					\draw [dashed] (1.4,-0.65) -- (2.1,-1.3);
					\draw (-2.1,-1.3) -- (-2.3,-1.95);
					\draw (-2.1,-1.3) -- (-1.8,-1.95);
					\draw (-0.7,-1.3) -- (-1.0,-1.95);
					\draw (-0.7,-1.3) -- (-0.5,-1.95);
					\draw (0.7,-1.3) -- (0.5,-1.95);
					\draw (0.7,-1.3) -- (1.0,-1.95);
					\draw (2.1,-1.3) -- (1.8,-1.95);
					\draw (2.1,-1.3) -- (2.3,-1.95);
					\draw (-2.3,-1.95) -- (-2.1,-2.6);
					\draw (-1.8,-1.95) -- (-2.1,-2.6);
					\draw (-1.0,-1.95) -- (-0.7,-2.6);
					\draw (-0.5,-1.95) -- (-0.7,-2.6);
					\draw (0.5,-1.95) -- (0.7,-2.6);
					\draw (1.0,-1.95) -- (0.7,-2.6);
					\draw (1.8,-1.95) -- (2.1,-2.6);
					\draw (2.3,-1.95) -- (2.1,-2.6);
					\draw [dashed] (-2.1,-2.6) -- (-1.4,-3.25);
					\draw [dashed] (-0.7,-2.6) -- (-1.4,-3.25);
					\draw [dashed] (0.7,-2.6) -- (1.4,-3.25);
					\draw [dashed] (2.1,-2.6) -- (1.4,-3.25);
					\draw (-1.4,-3.25) -- (0,-3.9);
					\draw (1.4,-3.25) -- (0,-3.9);
					
				\end{tikzpicture}
				\caption{We took a complete binary tree. We glued it to its copy.}
				\label{fig:tree_a}
			\end{subfigure}
			\hfill
			\begin{subfigure}[t]{0.45\textwidth}
				\centering
				\includegraphics[width=\textwidth]{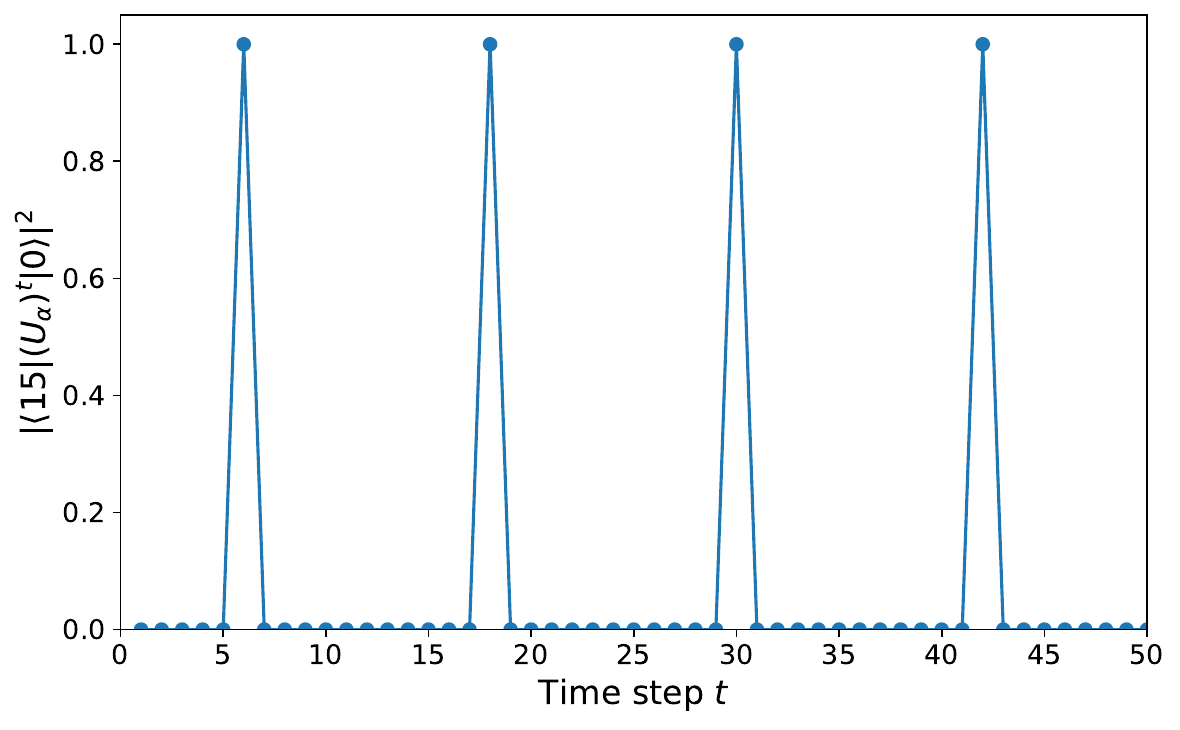}
				\caption{This is a time versus fidelity plot from vertices $r_1 = 0$ to $r_2 = 15$ in the graph \autoref{fig:tree_a} with $\alpha=0.5$. PST occurs from $r_1 = 0$ to $r_2 = 15$ at $t=6, 18, 30 \ldots$.}
			\end{subfigure}
			
			\vspace{0.2cm}
			
			\begin{subfigure}[t]{0.48\textwidth}
				\centering
				\begin{tikzpicture}[every node/.style={font=\footnotesize, inner sep=1.5pt}]

					\draw [fill] (0,0) circle [radius=0.05];
					\node [above] at (0,0) {$0$};
					
					\draw [fill] (-1.4,-0.65) circle [radius=0.05];
					\node [above] at (-1.4,-0.65) {$1$};
					\draw [fill] (1.4,-0.65) circle [radius=0.05];
					\node [above] at (1.4,-0.65) {$2$};
					
					\draw [fill] (-2.1,-1.3) circle [radius=0.05];
					\node [above] at (-2.1,-1.3) {$3$};
					\draw [fill] (-0.7,-1.3) circle [radius=0.05];
					\node [above] at (-0.7,-1.3) {$4$};
					\draw [fill] (0.7,-1.3) circle [radius=0.05];
					\node [above] at (0.7,-1.3) {$5$};
					\draw [fill] (2.1,-1.3) circle [radius=0.05];
					\node [above] at (2.1,-1.3) {$6$};
					
					\draw [fill] (-2.3,-1.95) circle [radius=0.05];
					\node [left] at (-2.3,-1.95) {$7$};
					\draw [fill] (-1.8,-1.95) circle [radius=0.05];
					\node [right] at (-1.8,-1.95) {$8$};
					\draw [fill] (-0.5,-1.95) circle [radius=0.05];
					\node [right] at (-0.5,-1.95) {$10$};
					\draw [fill] (0.5,-1.95) circle [radius=0.05];
					\node [left] at (0.5,-1.95) {$11$};
					\draw [fill] (1.0,-1.95) circle [radius=0.05];
					\node [right] at (1.0,-1.95) {$12$};
					\draw [fill] (1.8,-1.95) circle [radius=0.05];
					\node [left] at (1.8,-1.95) {$13$};
					\draw [fill] (2.3,-1.95) circle [radius=0.05];
					\node [right] at (2.3,-1.95) {$14$};
					
					\draw [fill] (-2.1,-2.6) circle [radius=0.05];
					\node [below] at (-2.1,-2.6) {$18$};
					\draw [fill] (-0.7,-2.6) circle [radius=0.05];
					\node [below] at (-0.7,-2.6) {$19$};
					\draw [fill] (0.7,-2.6) circle [radius=0.05];
					\node [below] at (0.7,-2.6) {$20$};
					\draw [fill] (2.1,-2.6) circle [radius=0.05];
					\node [below] at (2.1,-2.6) {$21$};
					
					\draw [fill] (-1.4,-3.25) circle [radius=0.05];
					\node [below] at (-1.4,-3.25) {$16$};
					\draw [fill] (1.4,-3.25) circle [radius=0.05];
					\node [below] at (1.4,-3.25) {$17$};
					
					\draw [fill] (0,-3.9) circle [radius=0.05];
					\node [below] at (0,-3.9) {$15$};

					\draw (0,0) -- (-1.4,-0.65);
					\draw (0,0) -- (1.4,-0.65);
					\draw [dashed] (-1.4,-0.65) -- (-2.1,-1.3);
					\draw [dashed] (-1.4,-0.65) -- (-0.7,-1.3);
					\draw [dashed] (1.4,-0.65) -- (0.7,-1.3);
					\draw [dashed] (1.4,-0.65) -- (2.1,-1.3);
					\draw (-2.1,-1.3) -- (-2.3,-1.95);
					\draw (-2.1,-1.3) -- (-1.8,-1.95);
					\draw (-0.7,-1.3) -- (-0.5,-1.95);
					\draw (0.7,-1.3) -- (0.5,-1.95);
					\draw (0.7,-1.3) -- (1.0,-1.95);
					\draw (2.1,-1.3) -- (1.8,-1.95);
					\draw (2.1,-1.3) -- (2.3,-1.95);
					\draw (-2.3,-1.95) -- (-2.1,-2.6);
					\draw (-1.8,-1.95) -- (-2.1,-2.6);
					\draw (-0.5,-1.95) -- (-0.7,-2.6);
					\draw (0.5,-1.95) -- (0.7,-2.6);
					\draw (1.0,-1.95) -- (0.7,-2.6);
					\draw (1.8,-1.95) -- (2.1,-2.6);
					\draw (2.3,-1.95) -- (2.1,-2.6);
					\draw [dashed] (-2.1,-2.6) -- (-1.4,-3.25);
					\draw [dashed] (-0.7,-2.6) -- (-1.4,-3.25);
					\draw [dashed] (0.7,-2.6) -- (1.4,-3.25);
					\draw [dashed] (2.1,-2.6) -- (1.4,-3.25);
					\draw (-1.4,-3.25) -- (0,-3.9);
					\draw (1.4,-3.25) -- (0,-3.9);
					
				\end{tikzpicture}
				\caption{Suppose the edge (4, 9) is damaged in \autoref{fig:tree_a}. Therefore, we remove (4, 9) and  (9, 19) to maintain the glued tree structure.}
				\label{fig:tree_b}
			\end{subfigure}
			\hfill
			\begin{subfigure}[t]{0.45\textwidth}
				\centering
				\includegraphics[width=\textwidth]{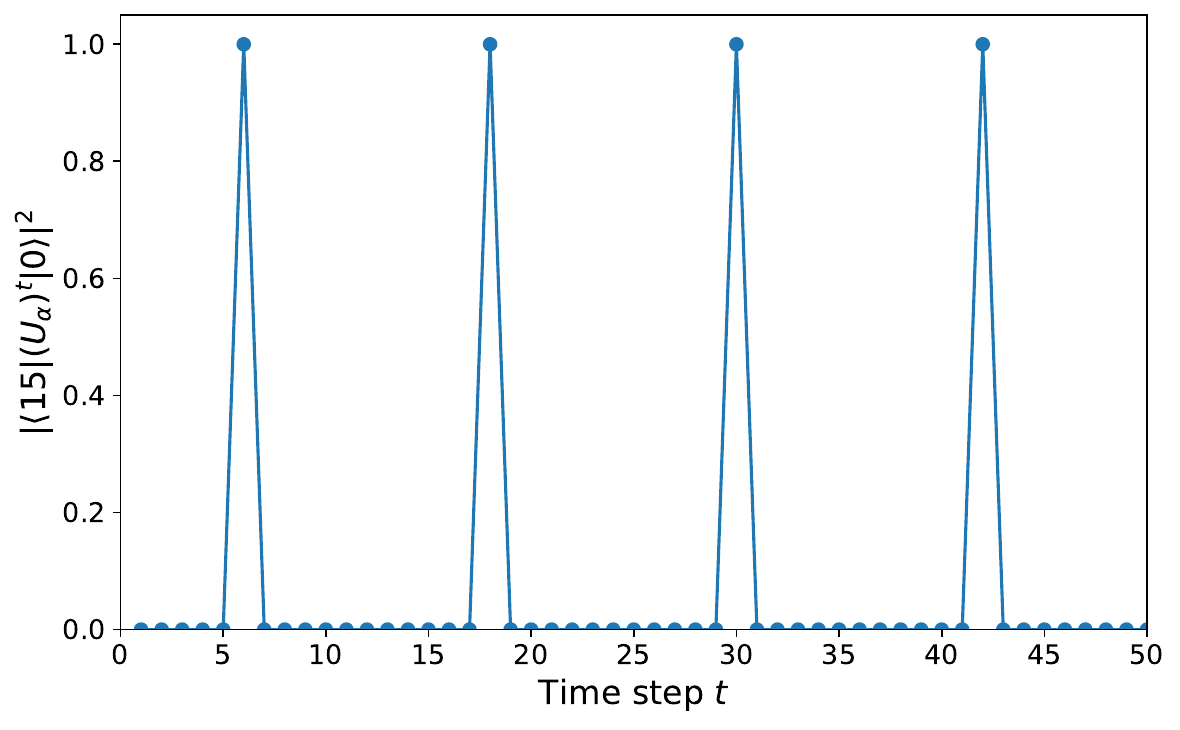}
				\caption{This is a time versus fidelity plot from vertices $r_1 = 0$ to $r_2=15$ in the graph \autoref{fig:tree_b} with $\alpha=0.5.$ PST occurs from $r_1=0$ to $r_2=15$ at $t=6, 18, 30 \ldots$.}
			\end{subfigure}
			
			\vspace{0.2cm}
			
			\begin{subfigure}[t]{0.48\textwidth}
				\centering
				\begin{tikzpicture}[every node/.style={font=\footnotesize, inner sep=1.5pt}]

					\draw [fill] (0,0) circle [radius=0.05];
					\node [above] at (0,0) {$0$};
					
					\draw [fill] (-1.4,-0.65) circle [radius=0.05];
					\node [above] at (-1.4,-0.65) {$1$};
					\draw [fill] (1.4,-0.65) circle [radius=0.05];
					\node [above] at (1.4,-0.65) {$2$};
					
					\draw [fill] (-2.1,-1.3) circle [radius=0.05];
					\node [above] at (-2.1,-1.3) {$3$};
					\draw [fill] (-0.7,-1.3) circle [radius=0.05];
					\node [above] at (-0.7,-1.3) {$4$};
					\draw [fill] (2.1,-1.3) circle [radius=0.05];
					\node [above] at (2.1,-1.3) {$6$};
					
					\draw [fill] (-2.3,-1.95) circle [radius=0.05];
					\node [left] at (-2.3,-1.95) {$7$};
					\draw [fill] (-1.8,-1.95) circle [radius=0.05];
					\node [right] at (-1.8,-1.95) {$8$};
					\draw [fill] (-1.0,-1.95) circle [radius=0.05];
					\node [left] at (-1.0,-1.95) {$9$};
					\draw [fill] (-0.5,-1.95) circle [radius=0.05];
					\node [right] at (-0.5,-1.95) {$10$};
					\draw [fill] (1.8,-1.95) circle [radius=0.05];
					\node [left] at (1.8,-1.95) {$13$};
					\draw [fill] (2.3,-1.95) circle [radius=0.05];
					\node [right] at (2.3,-1.95) {$14$};
					
					\draw [fill] (-2.1,-2.6) circle [radius=0.05];
					\node [below] at (-2.1,-2.6) {$18$};
					\draw [fill] (-0.7,-2.6) circle [radius=0.05];
					\node [below] at (-0.7,-2.6) {$19$};
					\draw [fill] (2.1,-2.6) circle [radius=0.05];
					\node [below] at (2.1,-2.6) {$21$};
					
					\draw [fill] (-1.4,-3.25) circle [radius=0.05];
					\node [below] at (-1.4,-3.25) {$16$};
					\draw [fill] (1.4,-3.25) circle [radius=0.05];
					\node [below] at (1.4,-3.25) {$17$};
					
					\draw [fill] (0,-3.9) circle [radius=0.05];
					\node [below] at (0,-3.9) {$15$};

					\draw (0,0) -- (-1.4,-0.65);
					\draw (0,0) -- (1.4,-0.65);
					\draw [dashed] (-1.4,-0.65) -- (-2.1,-1.3);
					\draw [dashed] (-1.4,-0.65) -- (-0.7,-1.3);
					\draw [dashed] (1.4,-0.65) -- (2.1,-1.3);
					\draw (-2.1,-1.3) -- (-2.3,-1.95);
					\draw (-2.1,-1.3) -- (-1.8,-1.95);
					\draw (-0.7,-1.3) -- (-1.0,-1.95);
					\draw (-0.7,-1.3) -- (-0.5,-1.95);
					\draw (2.1,-1.3) -- (1.8,-1.95);
					\draw (2.1,-1.3) -- (2.3,-1.95);
					\draw (-2.3,-1.95) -- (-2.1,-2.6);
					\draw (-1.8,-1.95) -- (-2.1,-2.6);
					\draw (-1.0,-1.95) -- (-0.7,-2.6);
					\draw (-0.5,-1.95) -- (-0.7,-2.6);
					\draw (1.8,-1.95) -- (2.1,-2.6);
					\draw (2.3,-1.95) -- (2.1,-2.6);
					\draw [dashed] (-2.1,-2.6) -- (-1.4,-3.25);
					\draw [dashed] (2.1,-2.6) -- (1.4,-3.25);
					\draw [dashed] (-1.4,-3.25) -- (-0.7,-2.6);
					\draw (-1.4,-3.25) -- (0,-3.9);
					\draw (1.4,-3.25) -- (0,-3.9);
					
				\end{tikzpicture}
				\caption{If the edge (2, 5) is damaged in \autoref{fig:tree_a}, we remove (2, 5), (5, 11), (5, 12), (11, 20),(12, 20) and (20,17) to maintain glued tree structure.}
				\label{fig:tree_c}
			\end{subfigure}
			\hfill
			\begin{subfigure}[t]{0.45\textwidth}
				\centering
				\includegraphics[width=\textwidth]{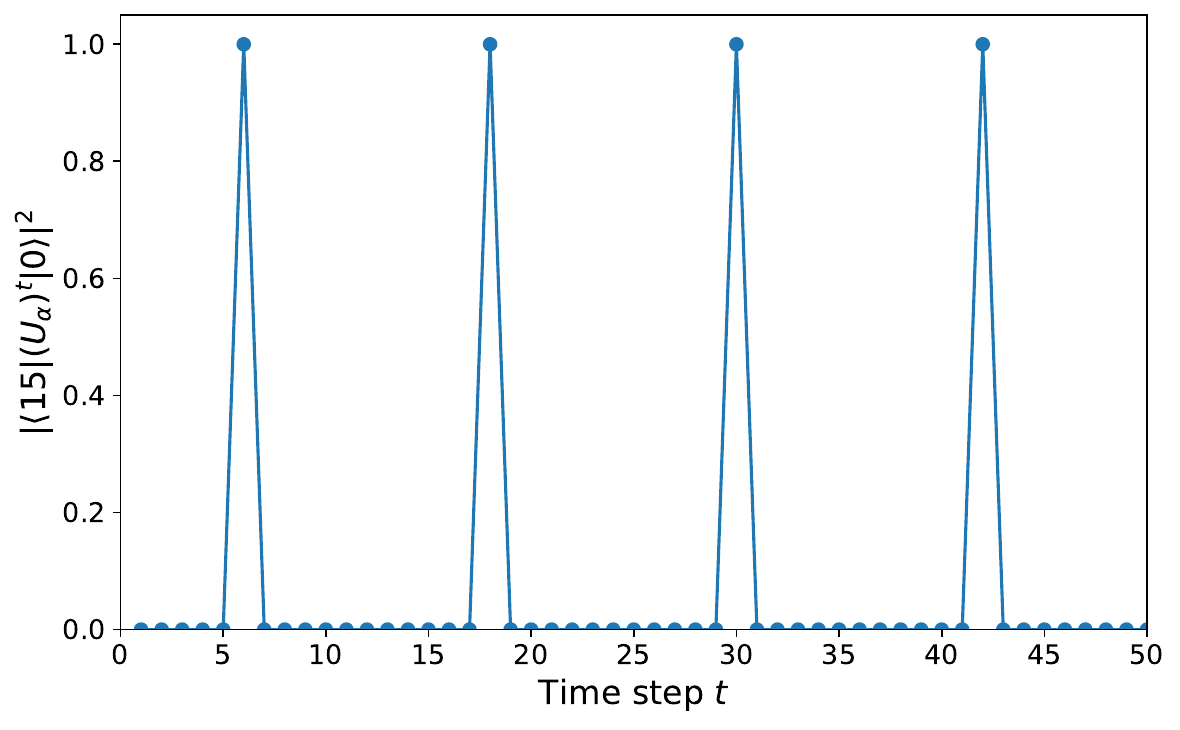}
				\caption{This is a time versus fidelity plot from vertices  $r_1=0$ to $r_2=15$ in the graph \autoref{fig:tree_c} with $\alpha=0.5.$ PST occurs from $r_1=0$ to $r_2=15$ at time step $t=6, 18, 30 \ldots$.}
			\end{subfigure}
			\caption{In the \autoref{fig:tree_a} we consider a glued tree. When an edge is damaged we need to remove it. We also remove a number of additional edges to maintain the structure of glued tree. The graphs in two steps of this operations are depicted in \autoref{fig:tree_b} and \autoref{fig:tree_c}. The resultants graphs also allow PST between $r_1 = 0$ and $r_2 = 15$.}
		\end{figure}
		
		\begin{observation}
			Making larger glued trees also support PST. For example consider \autoref{fig:tree_a}.
		\end{observation}

		\begin{observation}
			Let an edge is damaged. Remove the edges, including the damaged edge, such that the new graph preserves the structure of the glued tree. The new graph also supports PST between the roots at same time step equal to distance between the roots.
		\end{observation}

		A rooted glued tree acts as a high-capacity, multi-path quantum router. It consists of an input tree $T_1$ and an output tree $T_2$, each of depth $d$, joined together at their $2^d$ leaf boundaries. As a result, we obtain a deterministic, measurement-free spatial transport fabric. As a quantum router, it is governed by the following architectural and physical mechanisms:
		\begin{enumerate}
			\item 
				\textbf{Source Port (Root $r_1$):} It acts as the input terminal into the network fabric.
			\item 
				\textbf{Routing Fabric (Hierarchical Branching Levels):} The intermediate levels create a distributed multi-path transmission corridor.
			\item 
				\textbf{Destination Port (Root $r_2$):} It serves as the output terminal where the distributed wavefront recombines constructively at a deterministic arrival time $\tau = 2d$.
		\end{enumerate}
		
		In an unsigned quantum walk on a tree, routing fails. The tree levels induce asymmetric scattering and back-reflections and trap the wave-packets in an exponentially broad distribution which, does not focus at the opposite root. A signed rooted glued tree resolves this problem through interference-driven guidance.  Assigning alternating edge signs $\sigma = +1$ and $\sigma = -1$ between successive generational levels establishes a uniform localized coin probability $p(\vec{e}) = 1/2$ at each branching junction. We observe numerically that back-scattering toward the parent nodes is canceled out at each step. The wavepacket travels strictly forward along all possible paths with constant ballistic group velocity. At the central interface where we glue the leaves of $T_1$ and $T_2$, the relative sign modulation ensures that the wave-packets cross the boundary without reflecting back into the input tree.

	\section{A switch in a routing network}
	
		\begin{definition}
			\textbf{(Signed dumbbell graph)} Let $\Sigma = (V(G), E(G), \sigma_G)$ be a signed graph where $G$ is the dumbbell graph $D_{p,k,q}$, obtained from two vertex-disjoint cycles $C_p$ and $C_q$ joined by a path $P_{k+2}$, $k \geq 0$, and let $\sigma_G : E(G) \to \{+,-\}$ be the corresponding sign function. We call $\Sigma$ a signed dumbbell graph. 
		\end{definition}
		
		When $k=0$, the connecting path reduces to a single edge, such that $G = D_{p,0,q}$. We refer to this edge as the bridge. In this paper, we focus on two signed dumbbell graphs, denoted $\Sigma_1$ and $\Sigma_2$, on the underlying topology $D_{2m,0,2n}$. In $\Sigma_1$ all edges are positive. In contrast, in $\Sigma_2$ all the edges in the cycles $C_{2m}$ and $C_{2n}$ are positive, but the bridge is negative. The graph $\Sigma_2$ is depicted in \autoref{joined cycle 4 6} with $m = 2$ and $n = 3$. In the case of $\Sigma_1$, the bridge edge $(2, 4)$ is also positive.
	
		The vertices in the dumbbell graphs $G = D_{2m,0,2n}$ can be labeled by $0, 1, \dots, m, m+1, \dots, 2m-1, 2m, \dots, 2m+2n-1$. The vertices $0, 1, \dots m, (m + 1), \dots (2m - 1)$ form a cycle of $2m$ vertices. The second cycle of $2n$ vertices is formed by the vertices $2m, (2m + 1), \dots, (2m + n), \dots (2m + 2n -1)$. All the edges of the two cycles are positive. The bridge $(m, 2m)$ may be positive or negative for $\Sigma_1$, or $\Sigma_2$, respectively. Note that the vertices $0$ and $(2m + n)$ have the maximum distance between them.

	\begin{figure}
		\centering
		\begin{subfigure}[t]{0.31\textwidth}
			\centering
			\begin{tikzpicture}[scale=0.3]
				\draw[fill] (-2,0) circle (0.08);
				\node[left] at (-2,0) {$0$};
				
				\draw[fill] (0,-2) circle (0.08);
				\node[below] at (0,-2) {$1$};
				
				\draw[fill] (2,0) circle (0.08);
				\node[above right] at (2,0) {$2$};
				
				\draw[fill] (0,2) circle (0.08);
				\node[above] at (0,2) {$3$};
				
				\draw (-2,0)--(0,-2)--(2,0)--(0,2)--cycle;
				
				\begin{scope}[shift={(8,0)}]
					\draw[fill] (-2,0) circle (0.08);
					\node[above left] at (-2,0) {$4$};
					
					\draw[fill] (-1,-1.732) circle (0.08);
					\node[below left] at (-1,-1.732) {$5$};
					
					\draw[fill] (1,-1.732) circle (0.08);
					\node[below right] at (1,-1.732) {$6$};
					
					\draw[fill] (2,0) circle (0.08);
					\node[right] at (2,0) {$7$};
					
					\draw[fill] (1,1.732) circle (0.08);
					\node[above right] at (1,1.732) {$8$};
					
					\draw[fill] (-1,1.732) circle (0.08);
					\node[above left] at (-1,1.732) {$9$};
					
					\draw (-2,0)--(-1,-1.732)--(1,-1.732)--(2,0)--(1,1.732)--(-1,1.732)--cycle;
				\end{scope}
				\draw[dashed] (2,0)--(6,0);
			\end{tikzpicture}
			\caption{The signed dumbbell graph $D_{4,0,6}$ where the edges in the cycles and all  positive, and the bridge is negative.}
			\label{joined cycle 4 6}
		\end{subfigure}
		\hspace{0.25cm}
		\begin{subfigure}[t]{0.31\textwidth}
			\centering
			\includegraphics[width=\textwidth]{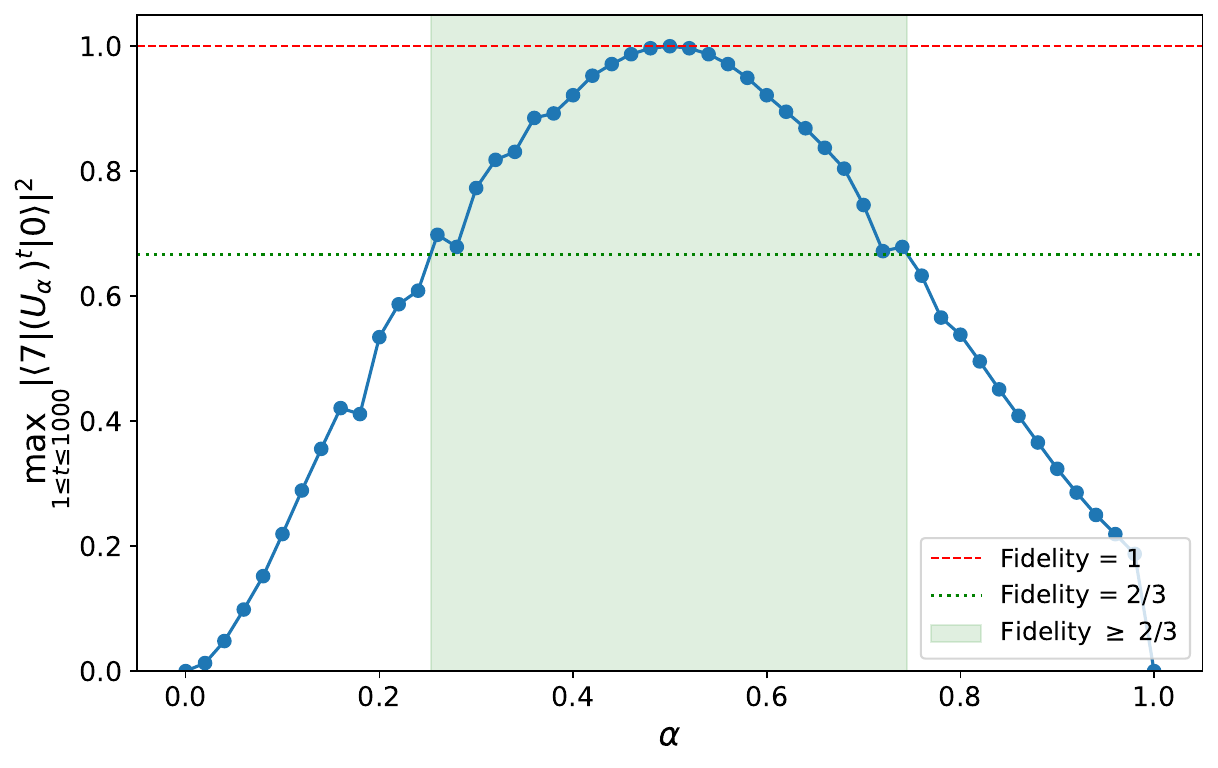}
			\caption{We plot the fidelity with respect to $\alpha$ to transmit information from vertex $0$ to $7$ of the graph in \autoref{joined cycle 4 6}. PST occurs only for $\alpha=0.5.$ The fidelity is more than $2/3$ for $0. 25 \le \alpha \le 0.74$.}
		\end{subfigure}
		\hspace{.25cm}
		\begin{subfigure}[t]{0.31\textwidth}
			\centering
			\includegraphics[width=\textwidth]{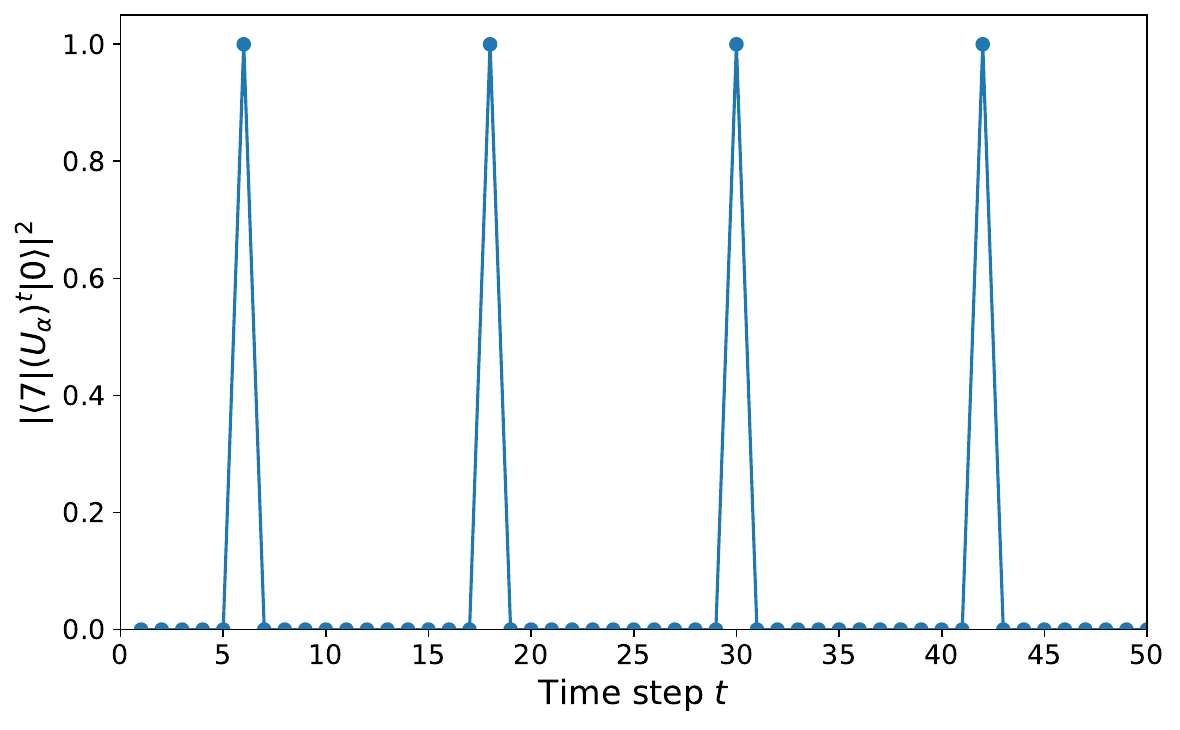}
			\caption{For $\alpha=0.5.$ we calculate fidelity from vertex $0$ to vertex $7$ of the graph in \autoref{joined cycle 4 6}. Note that PST occurs for $t = 6, 18, 30$, etc.}
		\end{subfigure}
		\caption{We study the fidelity of state transfer from vertex $0$ to vertex $7$ of the graph in \autoref{joined cycle 4 6} with respect to $\alpha$ and time.}
	\end{figure}
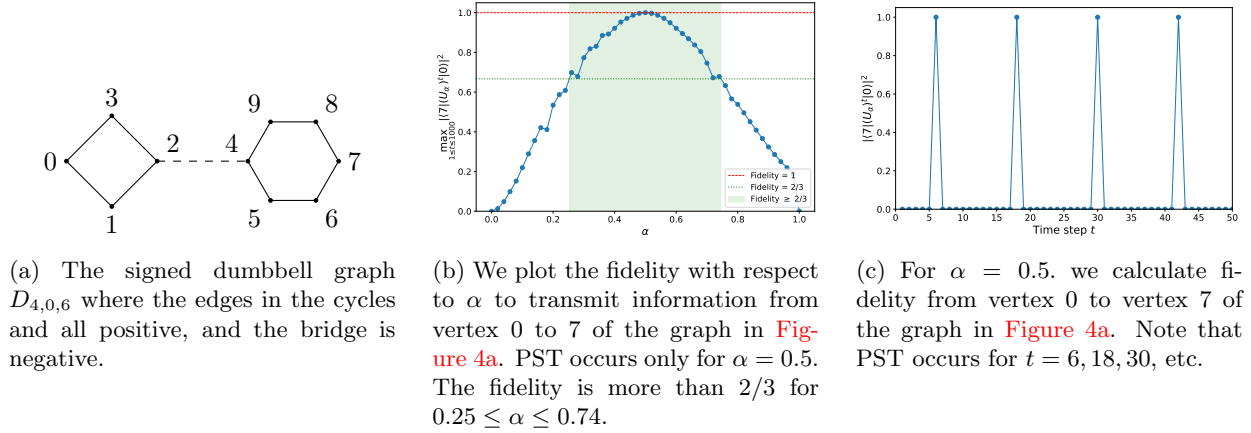

	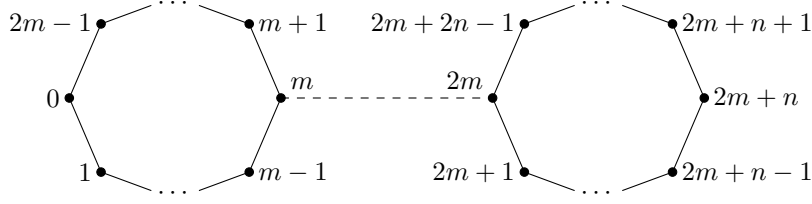
\begin{figure}
		\centering
		\begin{tikzpicture}[scale=0.7]
			\draw[fill] (-2,0) circle (0.08);
			\node[left] at (-2,0) {$0$};
			
			\draw[fill] (-1.4,-1.4) circle (0.08);
			\node[left] at (-1.4,-1.4) {$1$};
			
			\node at (0,-1.8) {$\cdots$};
			
			\draw[fill] (1.4,-1.4) circle (0.08);
			\node[right] at (1.4,-1.4) {$m-1$};
			
			\draw[fill] (2,0) circle (0.08);
			\node[above right ] at (2,0) {$m$};
			
			\draw[fill] (1.4,1.4) circle (0.08);
			\node[right] at (1.4,1.4) {$m+1$};
			
			\node at (0,1.8) {$\cdots$};
			
			\draw[fill] (-1.4,1.4) circle (0.08);
			\node[left] at (-1.4,1.4) {$2m-1$};
			
			\draw
			(-2,0)--(-1.4,-1.4)
			(-1.4,-1.4)--(-0.45,-1.75)
			(0.45,-1.75)--(1.4,-1.4)
			(1.4,-1.4)--(2,0)
			(2,0)--(1.4,1.4)
			(1.4,1.4)--(0.45,1.75)
			(-0.45,1.75)--(-1.4,1.4)
			(-1.4,1.4)--(-2,0);
			
			\begin{scope}[shift={(8,0)}]
				
				\draw[fill] (-2,0) circle (0.08);
				\node[above left] at (-2,0) {$2m$};
				
				\draw[fill] (-1.4,-1.4) circle (0.08);
				\node[left] at (-1.4,-1.4) {$2m+1$};
				
				\node at (0,-1.8) {$\cdots$};
				
				\draw[fill] (1.4,-1.4) circle (0.08);
				\node[right] at (1.4,-1.4) {$2m+n-1$};
				
				\draw[fill] (2,0) circle (0.08);
				\node[right] at (2,0) {$2m+n$};
				
				\draw[fill] (1.4,1.4) circle (0.08);
				\node[right] at (1.4,1.4) {$2m+n+1$};
				
				\node at (0,1.8) {$\cdots$};
				
				\draw[fill] (-1.4,1.4) circle (0.08);
				\node[left] at (-1.4,1.4) {$2m+2n-1$};
				
				\draw
				(-2,0)--(-1.4,-1.4)
				(-1.4,-1.4)--(-0.45,-1.75)
				(0.45,-1.75)--(1.4,-1.4)
				(1.4,-1.4)--(2,0)
				(2,0)--(1.4,1.4)
				(1.4,1.4)--(0.45,1.75)
				(-0.45,1.75)--(-1.4,1.4)
				(-1.4,1.4)--(-2,0);
				
			\end{scope}
			\draw[dashed] (2,0)--(6,0);
		\end{tikzpicture}
		\caption{A general dumbbell graph $D_{2m, 0, 2n}$ which is described in \autoref{effect of sign thm}. PST occurs between $0$ and $2m+n$ at time $t=m+1+n$.}
		\label{joined cycle}
	\end{figure}
	
	\begin{theorem}\label{effect of sign thm}
		For $\alpha = 0.5$, the graph $\Sigma_2$ admits PST between $0$ and $2m + n$ at time $t = m + 1 + n$, whereas $\Sigma_1$ does not admit PST.
	\end{theorem}
	
	\begin{proof}
		All vertices of $\Sigma_2$ except vertices $m$ and $2m$ have degree $2$. From equation (\ref{prob}), we get $p\overrightarrow{(m,m+1)}=\frac{0.5}{2}=0.25,$ $p\overrightarrow{(m,m-1)}=\frac{0.5}{2}=0.25$ and $p\overrightarrow{(m,2m)}=\frac{0.5}{1}=0.5.$ Again from equation (\ref{prob}) we get, $p\overrightarrow{(2m,2m+2n-1)}=\frac{0.5}{2}=0.25,$ $p\overrightarrow{(2m,2m+1)}=\frac{0.5}{2}=0.25,$ and $p\overrightarrow{(2m,m)}=\frac{0.5}{1}=0.5.$ Now, we have the following calculation:
		\begin{align*}
			& (U_{0.5})^{(m+1+n)}\ket{0}\\ &=(U_{0.5})^{(m+1+n)}{\frac{1}{\sqrt{2}}(\ket{\overrightarrow{(0,1)}} + \ket{\overrightarrow{(0,2m-1)}})} 
			\\
			&=(U_{0.5})^{((m-1)+1+n)}\frac{1}{\sqrt{2}}(U_{0.5}\ket{\overrightarrow{(0,1)}} + U_{0.5}\ket{\overrightarrow{(0,2m-1)}}) \\
			&=(U_{0.5})^{((m-2)+1+n)}\frac{1}{\sqrt{2}}(U_{0.5}\ket{\overrightarrow{(2m-1,0)}} + U_{0.5}\ket{\overrightarrow{(1,0)}})
			\hspace{0.25cm}\parbox[t]{0.40\linewidth}{%
				(From the lemma \ref{backward moving} and the fact that $N(0)=\{1,2m-1\}$.)
			}\\
			&=(U_{0.5})^{((m-3)+1+n)}\frac{1}{\sqrt{2}}(U_{0.5}\ket{\overrightarrow{(2m-2,2m-1)}} + U_{0.5}\ket{\overrightarrow{(2,1)}})
			\hspace{0.25cm}\hspace{0.25cm}\parbox[t]{0.34\linewidth}{%
				(From the lemma \ref{backward moving} and the facts that $N(2m-1)=\{0,2m-2\}$ and $N(1)=\{0,2\}$.)
			}\\
			&=(U_{0.5})^{((m-4)+1+n)}\frac{1}{\sqrt{2}}(U_{0.5}\ket{\overrightarrow{(2m-3,2m-2)}} + (U_{0.5})\ket{\overrightarrow{(3,2)}})
			\hspace{0.25cm}\parbox[t]{0.34\linewidth}{%
				(From the lemma \ref{backward moving} and the facts that $N(2m-2)=\{2m-1,2m-3\}$ and $N(2)=\{1,3\}$.)
			}\\
			&\hspace{3cm}\vdots\\
			&=(U_{0.5})^{((m-m)+1+n)}\frac{1}{\sqrt{2}}(U_{0.5}\ket{\overrightarrow{(2m-(m-1),2m-(m-2))}} \\
			&\qquad{}+ U_{0.5}\ket{\overrightarrow{(m-1,m-2)}})\\
			&=(U_{0.5})^{(0+1+n)}\frac{1}{\sqrt{2}}(U_{0.5}\ket{\overrightarrow{(m+1,m+2)}} + U_{0.5}\ket{\overrightarrow{(m-1,m-2)}})\\
			&=(U_{0.5})^n\frac{1}{\sqrt{2}}(U_{0.5}\ket{\overrightarrow{(m,m+1)}} + U_{0.5}\ket{\overrightarrow{(m,m-1)}})
			\hspace{0.25cm}\parbox[t]{0.40\linewidth}{%
				(from $N(m+1)=\{m,m+2\}$, $N(m-1)=\{m,m-2\}$ and lemma \ref{backward moving})
			}\\
			&=(U_{0.5})^n\frac{1}{\sqrt{2}}((((U_{0.5})_{0.5})_{\overrightarrow{(m+1,m)},\overrightarrow{(m,m+1)}}\ket{\overrightarrow{(m+1,m)}} + (U_{0.5})_{\overrightarrow{(m-1,m)},\overrightarrow{(m,m+1)}}\ket{\overrightarrow{(m-1,m)}} \\
			&\qquad {} + (U_{0.5})_{\overrightarrow{(2m,m)},\overrightarrow{(m,m+1)}}\ket{\overrightarrow{(2m,m)}})+ ((U_{0.5})_{\overrightarrow{(m+1,m)},\overrightarrow{(m,m-1)}}\ket{\overrightarrow{(m+1,m)}}  \\
			&\qquad {} +(U_{0.5})_{\overrightarrow{(m-1,m)},\overrightarrow{(m,m-1)}}\ket{\overrightarrow{(m-1,m)}} + (U_{0.5})_{\overrightarrow{(2m,m)},\overrightarrow{(m,m-1)}}\ket{\overrightarrow{(2m,m)}}))
			\hspace{0.25cm}\parbox[t]{0.22\linewidth}{%
				(from $N(m)=\{m+1,m-1,2m\}$ and equation \ref{evolution}.)
			}\\
			&=(U_{0.5})^n\frac{1}{\sqrt{2}}((2\sqrt{0.25}\sqrt{0.25}-1)\ket{\overrightarrow{(m+1,m)}}
			+ (2\sqrt{0.25}\sqrt{0.25})\ket{\overrightarrow{(m-1,m)}}\\	
			&\qquad {} + (2\sqrt{0.25}\sqrt{0.5})\ket{\overrightarrow{(2m,m)}} + (2\sqrt{0.25}\sqrt{0.25})\ket{\overrightarrow{(m+1,m)}} \\
			&\qquad{}+ (2\sqrt{0.25}\sqrt{0.25}-1)\ket{\overrightarrow{(m-1,m)}} + (2\sqrt{0.25}\sqrt{0.5})\ket{\overrightarrow{(2m,m)}}) 
			\hspace{0.25cm}\parbox[t]{0.25\linewidth}{%
				(from equations \ref{evolution} and \ref{szegedy matrix}.)
			}\\
			&=(U_{0.5})^n\frac{1}{\sqrt{2}}(-0.5\ket{\overrightarrow{(m+1,m)}}+ 0.5\ket{\overrightarrow{(m-1,m)}}+\sqrt{0.5}\ket{\overrightarrow{(2m,m)}}\\
			&\qquad{}+0.5\ket{\overrightarrow{(m+1,m)}}-0.5\ket{\overrightarrow{(m-1,m)}}+\sqrt{0.5}\ket{\overrightarrow{(2m,m)}})\\
			&=(U_{0.5})^n(0.5\ket{\overrightarrow{(2m,m)}}+0.5\ket{\overrightarrow{(2m,m)}})\\
			&=(U_{0.5})^n\ket{\overrightarrow{(2m,m)}}\\
			&=(U_{0.5})^{n-1}(U_{0.5}\ket{\overrightarrow{(2m,m)}}) \\
			&=(U_{0.5})^{n-1}((U_{0.5})_{\overrightarrow{(m,2m)}, \overrightarrow{(2m,m)}}\ket{\overrightarrow{(m,2m)}}+(U_{0.5})_{\overrightarrow{(2m+1,2m)},\overrightarrow{(2m,m)}}\ket{\overrightarrow{(2m+1,2m)}}\\
			&\qquad{}+(U_{0.5})_{\overrightarrow{(2m+2n-1,2m)},\overrightarrow{(2m,m)}}\ket{\overrightarrow{(2m+2n-1,2m)}})
			\hspace{0.5cm}\text{(from equation \ref{evolution} and \ref{szegedy matrix}.)} \\
			&=(U_{0.5})^{n-1}(0\ket{\overrightarrow{(m,2m)}}+(2\sqrt{0.5}\sqrt{0.25})\ket{\overrightarrow{(2m+1,2m)}}\\
			&\qquad{}+(2\sqrt{0.5}\sqrt{0.25})\ket{\overrightarrow{(2m+2n-1,2m)}})
			\hspace{0.5cm}\text{(from lemma \ref{half lemma} and  equation \ref{szegedy matrix}.)} \\
			&=(U_{0.5})^{n-1}\frac{1}{\sqrt{2}}(\ket{\overrightarrow{(2m+1,2m)}} + \ket{\overrightarrow{(2m+2n-1,2m)}})\\
			&=(U_{0.5})^{n-2}\frac{1}{\sqrt{2}}(U_{0.5}\ket{\overrightarrow{(2m+1,2m)}}+ U_{0.5}\ket{\overrightarrow{(2m+2n-1,2m)}})\\
			&=(U_{0.5})^{n-3}\frac{1}{\sqrt{2}}(U_{0.5}\ket{\overrightarrow{(2m+2,2m+1)}}+ U_{0.5}\ket{\overrightarrow{(2m+2n-2,2m+2n-1)}})
			\hspace{0.5cm}\parbox[t]{0.20\linewidth}{%
				(from $N(2m+1)=\{2m, 2m+2\}$,
				$N(2m+2n-1)=\{2m, 2m+2n-2\}$,\\
				and Lemma~\ref{backward moving}.)
			}\\
			&=(U_{0.5})^{n-4}\frac{1}{\sqrt{2}}(U_{0.5}\ket{\overrightarrow{(2m+3,2m+2)}}+ U_{0.5}\ket{\overrightarrow{(2m+2n-3,2m+2n-2)}})
			\hspace{0.5cm}\parbox[t]{0.2\linewidth}{%
				(from $N(2m+2)=\{2m+1, 2m+3\}$,
				$N(2m+2n-2)=\{2m+2n-1, 2m+2n-3\}$,
				and Lemma~\ref{backward moving}.)
			}\\
			&\hspace{3cm}\vdots\\
			&=(U_{0.5})^{n-n}\frac{1}{\sqrt{2}}(U_{0.5}\ket{\overrightarrow{(2m+n-1,2m+n-2)}} \\
			&\qquad{}+ U_{0.5}\ket{\overrightarrow{(2m+2n-(n-1),2m+2n-(n-2))}})\\
			&=\frac{1}{\sqrt{2}}(U_{0.5}\ket{\overrightarrow{(2m+n-1,2m+n-2)}}+ U_{0.5}\ket{\overrightarrow{(2m+n+1,2m+n+2)}})\\
			&=\frac{1}{\sqrt{2}}(\ket{\overrightarrow{(2m+n,2m+n-1)}}+ \ket{\overrightarrow{(2m+n,2m+n+1)}})
			\hspace{0.5cm}\parbox[t]{0.35\linewidth}{%
				(from $N(2m+n-1)=\{2m+n-2, 2m+n\}$,
				$N(2m+n+1)=\{2m+n+2, 2m+n\}$,
				and Lemma~\ref{backward moving}.)
			} \\
			&=\ket{2m+n}.
		\end{align*}	
		Therefore, there is a PST between the vertices $0$ and $(2m + n)$. In case of $\Sigma_1$, we have $p\overrightarrow{(m,m+1)} = p\overrightarrow{(m,m-1)} = p\overrightarrow{(m,2m)} = \frac{1}{3}$. Also, $p\overrightarrow{(2m,2m+2n-1)} = p\overrightarrow{(2m,2m+1)} = p\overrightarrow{(2m,m)} = \frac{1}{3}$. It leads to the fact that $(U_{0.5})^{(m+1+n)}\ket{0} \neq \ket{2m+n}$. Hence, in the graph $\Sigma_1$ there is no PST between $0$ and $(2m + n)$.
	\end{proof}
	
	Now we construct a topological quantum router with our dumbbell graph $D_{2m, 0, 2n}$. \autoref{effect of sign thm} states that $\Sigma_2$ allows PST between the extreme vertices $0$ and $(2m + n)$ when the bridge is negative. Place the source port $S$ at the input vertex $0$ located on the primary even cycle $C_{2m}$. The cycle $C_{2m}$ acts as an on-chip delay or memory loop. The exit port $D_{\text{exit}}$ is located at the antipodal vertex $(2m + n)$ on the cycle $C_{2n}$. The bridge arc $e_b = (m, 2m)$ may be positive or negative. 
	
	When the parameter of the unitary operator $\alpha = 0.5$, an arbitrary single-qubit quantum state $\vert{}\psi\rangle$ injected at port $0$ splits into two counter-propagating symmetric wavefronts: one clockwise along $(0, 1, \dots, m-1)$, and one counter-clockwise along $(0, 2m-1, \dots, m+1)$. Because the cycle has $2m$ vertices, both the paths to vertex $m$ have distance $m$. At time $t = m$, both wave-fronts arrive at junction node $m$ simultaneously. The routing decision is controlled by the sign $\sigma(e_b)$ in the following modes:
	
	\begin{enumerate}
		\item 
		\textbf{Optical Delay Mode ($\sigma(e_b) = +1$)}: 
		As the bridge is positive, vertex $m$ has three  positive edges and zero negative edges. The probabilities become $p = \frac{1}{3} \ne \frac{1}{2}$. From \autoref{half lemma}, the reflection-free condition $(U_\alpha)_{\vec{e}^{-1}, \vec{e}} = 0$ is violated ($2p - 1 = -\frac{1}{3} \ne 0$). Consequently, the state scatters and largely reflects back into the $2m$-cycle.
		
		\item 
		\textbf{Transmission Mode ($\sigma(e_b) = -1$)}: Now, the junction $m$ is incident to two positive edges and one negative bridge edge. The transition probabilities become $p(m, m+1) = 0.25$, $p(m, m-1) = 0.25$, and $p(m, 2m) = 0.5$. The wavepacket enters the $2n$-vertex cycle. It splits symmetrically around the ring, and refocuses at the antipodal terminal vertex $(2m+n)$ at time $\tau = m + 1 + n$ with unit fidelity $F = 1.0$.
	\end{enumerate}
	
	This architecture satisfies the criteria for a topological quantum router because its performance is enforced by global, discrete topological invariants. The routing mechanism does not require the cycles to be identical in size. No measurement is involved. The routing switch is executed purely by a discrete sign change $\sigma(e_b) \in \{+1, -1\}$.

	\section{Quantum communication in a noisy environment}
		
		The fundamental objective in this study is to understand how far an unknown quantum state can travel across our signed routing topologies before environmental noise destroys its quantum state fidelity. We assume that the quantum noise is applied on the state during each step of hopping. At each discrete time step $t \to t+1$, the walk undergoes a unitary coin-shift evolution $U$, followed by the noise superoperator $\mathcal{E}_\lambda$ \cite{kumar2018non}. Therefore, the state at time $(t + 1)$ can be written as 
		\begin{equation}
			\rho(t+1) = \mathcal{E}_\lambda\left( U \rho(t) U^\dagger \right) = \sum_k E_k \left( U \rho(t) U^\dagger \right) E_k^\dagger.
		\end{equation}
	
		\subsection{Amplitude damping Noise}
		
			In an integrated optical waveguide circuit, the state $\vert{}\vec{e}\rangle$ corresponds to a single-photon wave-packet traveling along a channel between two couplers.   The optical materials have an intrinsic absorption and scattering loss per unit distance. As the photon propagates from step to step, there is a probability $\lambda \in [0, 1]$ that the photon is absorbed or scattered out.   
			
			The Kraus operator representation of the generalized amplitude damping channel across an $N = 2|E(\overrightarrow{G})|$ dimensional arc Hilbert space is given by \cite{dutta2023qudit}:
			$$E_0 = \vert{}0\rangle\langle 0\vert{} + \sqrt{1 - \lambda} \sum_{i=1}^{N-1} \vert{}i\rangle\langle i\vert{}, \quad E_i = \sqrt{\lambda} \vert{}0\rangle\langle i\vert{} \quad (i = 1, \dots, N-1),$$
			where $\vert{}0\rangle$ represents the vacuum state.   Because $E_0$ damps all active orthogonal arc states by $\sqrt{1 - \lambda}$, the coherent population remaining in the computational arc subspace after one step is multiplied by $(1 - \lambda)$. After $t$ time steps, the fidelity to the target state scales as:
			$$F(t; \lambda) = (1 - \lambda)^t \cdot F_{\text{unitary}}(t)$$
			The perfect state transfer occurs at a specific, discrete arrival time $t = \tau$ with $F_{\text{unitary}}(\tau) = 1$. Therefore, $F(\tau; \lambda) = (1 - \lambda)^\tau.$ In case of glued signed tree $\tau = 2d$, where $d$ is the depth of the tree. Therefore, the delivered end-to-end fidelity at the target port scales as:
			$$F_{r_1 \rightarrow r_2}(\tau; \lambda) = (1 - \lambda)^{2d}.$$ 
			Also, for signed dumbbell graph $\tau = m + 1 + n$. With per-hop loss rate $\lambda$, the delivered state fidelity at the destination node is:
			$$F_{0 \to 2m+n}(\tau; \lambda) = (1 - \lambda)^{m + 1 + n}.$$
			
			In quantum networking, sending a degraded quantum state is only useful if it outperforms classical communication. If we try to transmit an unknown quantum state using only classical measurement and communication, the maximum achievable fidelity is bounded by $F_{\text{classical}} = \frac{2}{3} \approx 0.667$. For the state transfer on our graph to remain strictly quantum, the fidelity must exceed this threshold. To derive the maximum network hop limit $\tau_{\max}$ we apply the inequality $(1 - \lambda)^\tau > \frac{2}{3}$. It leads us to $\tau < \frac{\ln(2/3)}{\ln(1 - \lambda)}$. For small loss rates ($\lambda \ll 1$), using the approximation $\ln(1 - \lambda) \approx -\lambda$. Therefore,
			$$\tau_{\max} \approx \frac{\vert{}\ln(2/3)\vert{}}{\lambda} \approx \frac{0.4055}{\lambda}.$$
			In modern laser-written waveguides, loss is around $\lambda \approx 0.02$.Therefore, $$\tau_{\max} \approx \frac{0.4055}{0.02} \approx 20 \text{ hops}.$$
			This means our signed router can route quantum information across the dumbbell graphs with a total travel distance $\tau = m + 1 + n \le 20$ hops. In case of glued tree $2d \leq 20$, which requires the depth of the tree to satisfy $d \le 10$. The relation between the arrival time $\tau$ and state fidelity is depicted in \autoref{ADC_Figure}.
			
			\begin{figure}
				\centering
				\includegraphics[scale = .5]{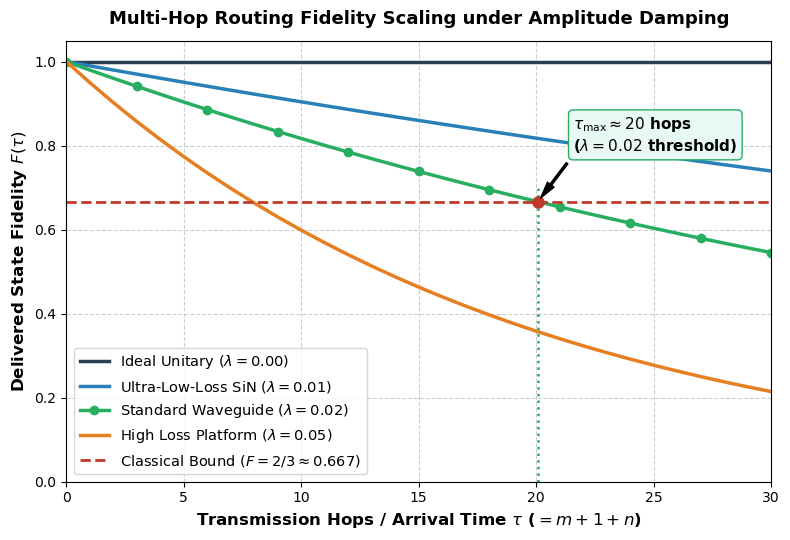}
				\caption{We plot the relation between the arrival time and state fidelity for different values of $\lambda$ in ADC channel.}
				\label{ADC_Figure}
			\end{figure}

		\subsection{Phase Damping Noise}
		
			Phase damping is physically distinct from amplitude damping. Amplitude damping causes photon leakage or energy loss. In contrast, phase damping preserves photon population completely, but destroys off-diagonal coherence. We use the Weyl-operator based Kraus operators for an $N = 2|E(G)|$ dimensional space:
			\begin{equation}
				E_{0,0} = \sqrt{1 - k} \, I; ~\text{and}~ E_{r,s} = \sqrt{\frac{k}{N^2 - 1}} \, U_{r,s} ~\text{for}~ ((r,s) \neq (0,0))
			\end{equation}
			Under pure dephasing that is $s = 0$, the Kraus operators reduce to
			\begin{equation}
				E_r = \frac{1}{\sqrt{N}} \sum_{m=0}^{N-1} e^{i \frac{2\pi}{N} r m} \vert{}m\rangle\langle m\vert{}
			\end{equation}
			When a state $\rho$ undergoes stochastic random phase fluctuations at each mode $\theta_m \sim \mathcal{N}(0, \sigma_\theta^2)$, we have $\langle m \vert{} \mathcal{E}(\rho) \vert{} n \rangle = \rho_{mn} \, \mathbb{E}\left[ e^{i(\theta_m - \theta_n)} \right]$. For independent Gaussian fluctuations $\mathbb{E}\left[ e^{i(\theta_m - \theta_n)} \right] = e^{-\sigma_\theta^2} \equiv (1 - p).$ While for diagonal terms ($m = n$), we have $\mathbb{E}\left[ e^{i(\theta_m - \theta_m)} \right] = 1$. Thus, in every single step
			\begin{equation}
				\rho_{ij}(1) = \begin{cases} \rho_{ii}(0), & i = j; \\ (1 - p)\rho_{ij}(0), & i \neq j; \end{cases}
			\end{equation}		
			where, $p \in [0, 1]$ is the dephasing probability per discrete time step along a waveguide channel or directional coupler. The noise channel is applied independently at each discrete hop $t = 1, 2, \dots, \tau$. Therefore, the off-diagonal elements evolve as:
			\begin{equation}
				\rho_{ij}(t) = (1 - p)^t \rho_{ij}(0) \quad (i \ne j).
			\end{equation}

			Now, consider the quantum walk on the dumbbell graph. The initial state is a symmetric coherent superposition of two counter-propagating paths: Clock-Wise (CW) and Counter-Clock-Wise (CCW). Therefore the initial state $\ket{\psi(0)}$ can be written as 
			$$\vert{}\psi(0)\rangle = \frac{1}{\sqrt{2}}\Big(\vert{}\text{CW}\rangle + \vert{}\text{CCW}\rangle\Big),$$
			with density matrix $\rho(0) = \ket{\psi(0)}\bra{\psi(0)}$. It has the off-diagonal coherence terms $\frac{1}{2}$. 
			
			At the destination node, when time $t = \tau$, perfect state transfer relies on constructive interference between these two paths. The interference cross-terms decay by a factor of $(1 - p)$ at every hop. Hence, the off-diagonal terms at time $\tau$ are given by
			$$\rho_{\text{CW}, \text{CCW}}(\tau) = \frac{1}{2}(1 - p)^\tau.$$
			Therefore, the delivered target state fidelity evaluates to:$$F_{\text{phase}}(\tau; p) = \langle\psi_{\text{target}}\vert{}\rho(\tau)\vert{}\psi_{\text{target}}\rangle = \frac{1}{2} + \frac{1}{2}(1 - p)^\tau.$$
			
			To beat the classical teleportation bound $F > 2/3$ we have $(1 - p)^\tau > \frac{1}{3}$, which indicates $\tau_{\max} = \frac{\ln(1/3)}{\ln(1 - p)} \approx \frac{1.0986}{p}$. 	For a realistic optical phase fluctuation rate $p = 0.02$, we obtain $\tau_{\max} \approx 55$ hops. Therefore, the phase damping preserves the quantum routing advantage across more than double the distance compared to photon loss ($\tau_{\max} \approx 20$ hops for amplitude damping). A similar derivation applies to the glued binary tree architecture.

	\section{Conclusion}
	
		In network physics, the routing schemes that  depend on fine-tuned path lengths or delicate coupling coefficients are fragile. Small manufacturing errors degrade optical interference and prevent the transmission of quantum information in the network. A topological router determines the routing path by discrete graph invariants. The arc-based Szegedy walk formulation on signed graphs admits a direct mapping onto integrated photonic waveguide circuits. The directed arc states $\ket{\vec{e}}$ represent discrete spatial waveguide modes carrying single-photon packets. Each vertex $u$ functions as a multi-mode optical directional coupler whose power division between positive and negative output ports is parameterized by the coin bias $\alpha$. A positive edge $\sigma(e) = +1$ corresponds to standard optical transmission, whereas a negative edge $\sigma(e) = -1$ is realized by introducing a localized $\pi$-phase shifter.

		Under this architecture, the sign-dependent perfect state transfer provides a mechanism for all-optical, measurement-free quantum routing. When the bridge connecting two even cycles carries a positive sign, destructive interference into the bridge mode is absent, trapping the quantum state in the primary cycle. The PST across the dumbbell does not depend on the specific sizes of the cycles. As long as both cycles are even-length bipartite cycles, the topological bipartite symmetry guarantees that the counter-propagating modes arrive at the junction exactly in phase. Another example of a topological router is the glued-tree architecture. In signed rooted glued trees, enforcing alternating edge signs across levels suppresses lateral leakage across sibling nodes through destructive interference. The walk behaves like a ballistically protected 1D channel connecting root vertices. As routing is governed by the combinatorial structure of graphs and discrete sign configuration rather than continuous parameter fine-tuning, this setup realizes a topologically protected, phase-controlled quantum router.
		
		This article presents all-unitary, coherent network routing protocols. No intermediate projective measurements or classical handshakes are performed. Consequently, transmission is non-destructive, preserving all superposition amplitudes and coherence phases intact. The routing selection is executed via a binary topological parameter which is the sign on the edges $\sigma \in \{+1, -1\}$ rather than continuous, analog parameter fine-tuning.
		
		Furthermore, our open-system noise analysis demonstrates that both architectures exhibit substantial tolerance to environmental decoherence. While photon absorption (amplitude damping) imposes an operational horizon of $\tau_{\max} \approx 0.4055/\lambda$ hops before falling to the classical benchmark ($F = 2/3$), pure dephasing (phase damping) retains a quantum routing advantage across more than double this distance ($\tau_{\max} \approx 1.0986/p$ hops). Because switching operations are driven entirely by binary topological gauge signs ($\sigma = \pm 1$) rather than continuous analog tuning, this paradigm provides an autonomous, fault-tolerant foundation for integrated photonic quantum routing fabrics.

	\section*{Acknowledgments} 
	
		The authors acknowledge the use of Gemini (Google) for language correction, scientific copy-editing, and manuscript structuring. All theoretical models, proofs, scientific interpretations, and final editorial revisions were thoroughly reviewed, verified, and approved by the authors, who take full responsibility for the contents and scientific integrity of this work.


\begin{thebibliography}{10}
		
		\bibitem{bose2003quantum}
		Sougato Bose.
		\newblock \href{https://arxiv.org/abs/quant-ph/0212041}{Quantum communication
			through an unmodulated spin chain}.
		\newblock {\em Physical review letters}, 91(20):207901, 2003.
		
		\bibitem{zueco2009quantum}
		David Zueco, Fernando Galve, Sigmund Kohler, and Peter H{\"a}nggi.
		\newblock \href{https://arxiv.org/abs/0905.4677}{Quantum router based on ac
			control of qubit chains}.
		\newblock {\em Physical Review A—Atomic, Molecular, and Optical Physics},
		80(4):042303, 2009.
		
		\bibitem{paganelli2013routing}
		Simone Paganelli, Salvatore Lorenzo, Tony~JG Apollaro, Francesco Plastina, and
		Gian~Luca Giorgi.
		\newblock \href{https://arxiv.org/abs/1301.5610}{Routing quantum information in
			spin chains}.
		\newblock {\em Physical Review A—Atomic, Molecular, and Optical Physics},
		87(6):062309, 2013.
		
		\bibitem{dutta2023quantum}
		Supriyo Dutta.
		\newblock \href{https://arxiv.org/abs/2302.10074}{Quantum routing in planar
			graph using perfect state transfer: S. Dutta}.
		\newblock {\em Quantum Information Processing}, 22(10):383, 2023.
		
		\bibitem{portugal2013quantum}
		Renato Portugal.
		\newblock {\em Quantum walks and search algorithms}.
		\newblock Springer Science \& Business Media, 2013.
		
		\bibitem{dutta2026perfect}
		Supriyo Dutta.
		\newblock Perfect state transfer using markovian quantum walk.
		\newblock {\em Annals of Physics}, 488:170411, 2026.
		
		\bibitem{zhan2014perfect}
		Xiang Zhan, Hao Qin, Zhi-hao Bian, Jian Li, and Peng Xue.
		\newblock \href{https://arxiv.org/abs/1405.6422}{Perfect state transfer and
			efficient quantum routing: A discrete-time quantum-walk approach}.
		\newblock {\em Physical Review A}, 90(1):012331, 2014.
		
		\bibitem{segawa2011localization}
		Etsuo Segawa.
		\newblock \href{https://arxiv.org/abs/1112.4982}{Localization of quantum walks
			induced by recurrence properties of random walks}.
		\newblock {\em arXiv preprint arXiv:1112.4982}, 2011.
		
		\bibitem{higuchi2017periodicity}
		Yusuke Higuchi, Norio Konno, Iwao Sato, and Etsuo Segawa.
		\newblock
		\href{https://www.jstage.jst.go.jp/article/iis/23/1/23_2017.A.10/_article}{Periodicity
			of the discrete-time quantum walk on a finite graph}.
		\newblock {\em Interdisciplinary Information Sciences}, 23(1):75--86, 2017.
		
		\bibitem{ozawa2019topological}
		Tomoki Ozawa, Hannah~M Price, Alberto Amo, Nathan Goldman, Mohammad Hafezi,
		Ling Lu, Mikael~C Rechtsman, David Schuster, Jonathan Simon, Oded Zilberberg,
		et~al.
		\newblock \href{https://arxiv.org/abs/1802.04173}{Topological photonics}.
		\newblock {\em Reviews of Modern Physics}, 91(1):015006, 2019.
		
		\bibitem{li2022multiport}
		Meng-Yu Li, Wen-Jie Chen, Xin-Tao He, and Jian-Wen Dong.
		\newblock
		\href{https://www.frontiersin.org/journals/physics/articles/10.3389/fphy.2022.902533/full}{Multiport
			Routing of Topologically Optical Transport Based on Merging of
			Valley-Dependent Edge States and Second-Order Corner States}.
		\newblock {\em Frontiers in Physics}, 10:902533, 2022.
		
		\bibitem{brown2012perfect}
		John Brown, Chris Godsil, Devlin Mallory, Abigail Raz, and Christino Tamon.
		\newblock \href{https://arxiv.org/abs/1211.0505}{Perfect state transfer on
			signed graphs}.
		\newblock {\em arXiv preprint arXiv:1211.0505}, 2012.
		
		\bibitem{harary1953notion}
		Frank Harary.
		\newblock On the notion of balance of a signed graph.
		\newblock {\em Michigan Mathematical Journal}, 2(2):143--146, 1953.
		
		\bibitem{west2001introduction}
		Douglas~Brent West et~al.
		\newblock {\em Introduction to graph theory}, volume~2.
		\newblock Prentice hall Upper Saddle River, 2001.
		
		\bibitem{wei2022towards}
		Shi-Hai Wei, Bo~Jing, Xue-Ying Zhang, Jin-Yu Liao, Chen-Zhi Yuan, Bo-Yu Fan,
		Chen Lyu, Dian-Li Zhou, You Wang, Guang-Wei Deng, et~al.
		\newblock Towards real-world quantum networks: a review.
		\newblock {\em Laser \& Photonics Reviews}, 16(3):2100219, 2022.
		
		\bibitem{azuma2023quantum}
		Koji Azuma, Sophia~E Economou, David Elkouss, Paul Hilaire, Liang Jiang,
		Hoi-Kwong Lo, and Ilan Tzitrin.
		\newblock Quantum repeaters: From quantum networks to the quantum internet.
		\newblock {\em Reviews of Modern Physics}, 95(4):045006, 2023.
		
		\bibitem{szegedy2004quantum}
		Mario Szegedy.
		\newblock Quantum speed-up of markov chain based algorithms.
		\newblock In {\em 45th Annual IEEE symposium on foundations of computer
			science}, pages 32--41. IEEE, 2004.
		
		\bibitem{szegedy2004spectra}
		Mario Szegedy.
		\newblock Spectra of quantized walks and a $\sqrt{\delta \varepsilon}$ rule.
		\newblock {\em arXiv preprint quant-ph/0401053}, 2004.
		
		\bibitem{christandl2004perfect}
		Matthias Christandl, Nilanjana Datta, Artur Ekert, and Andrew~J Landahl.
		\newblock Perfect state transfer in quantum spin networks.
		\newblock {\em Physical review letters}, 92(18):187902, 2004.
		
		\bibitem{chapman2016experimental}
		Robert~J Chapman, Matteo Santandrea, Zixin Huang, Giacomo Corrielli, Andrea
		Crespi, Man-Hong Yung, Roberto Osellame, and Alberto Peruzzo.
		\newblock Experimental perfect state transfer of an entangled photonic qubit.
		\newblock {\em Nature communications}, 7(1):11339, 2016.
		
		\bibitem{kumar2018non}
		N~Pradeep Kumar, Subhashish Banerjee, R~Srikanth, Vinayak Jagadish, and
		Francesco Petruccione.
		\newblock Non-markovian evolution: a quantum walk perspective.
		\newblock {\em Open Systems \& Information Dynamics}, 25(03):1850014, 2018.
		
		\bibitem{dutta2023qudit}
		Supriyo Dutta, Subhashish Banerjee, and Monika Rani.
		\newblock Qudit states in noisy quantum channels.
		\newblock {\em Physica Scripta}, 98(11):115113, 2023.
		
	\end{thebibliography}

\end{document}